\documentclass[11pt,twoside]{article}

\usepackage{fullpage}

\usepackage{epsf}
\usepackage{fancyhdr}
\usepackage{graphics}
\usepackage{graphicx}
\usepackage{psfrag}
\usepackage{microtype}
\usepackage{subfigure}
\usepackage{algorithmic}

\usepackage[linesnumbered,ruled]{algorithm2e}
\DontPrintSemicolon
\usepackage{color}

\usepackage{amsthm}
\usepackage{amsfonts}
\usepackage{amsmath}
\usepackage{amssymb,bbm}
\usepackage{natbib}

\usepackage[colorlinks,linkcolor=magenta,citecolor=blue, pagebackref=true,backref=true]{hyperref}
\renewcommand*{\backrefalt}[4]{%
    \ifcase #1 \footnotesize{(Not cited.)}%
    \or        \footnotesize{(Cited on page~#2.)}%
    \else      \footnotesize{(Cited on pages~#2.)}%
    \fi}

\long\def\comment#1{}
\usepackage{nicefrac}

\usepackage{chngpage}

 \usepackage{tabularx}%

\usepackage{enumitem}
\usepackage{booktabs}
\usepackage{caption}

\usepackage{mathtools}

\usepackage{fullpage}

\newtheorem{theorem}{Theorem}[section]
\newtheorem{corollary}[theorem]{Corollary}
\newtheorem{lemma}[theorem]{Lemma}
\newtheorem{proposition}[theorem]{Proposition}

\newtheorem{remark}[theorem]{Remark}

\newcommand{\st}{\textnormal{s.t.}}

\newcommand{\diag}{\textnormal{diag}}

\newcommand{\argmin}{\mathop{\rm argmin}}
\newcommand{\argmax}{\mathop{\rm argmax}}
\newcommand{\LCal}{\mathcal{L}}

\newcommand{\br}{\mathbb{R}}

\newcommand{\ba}{\begin{array}}
\newcommand{\ea}{\end{array}}

\newcommand{\Rspace}{\mathbb{R}}
\newcommand{\one}{\textbf{1}}
\newcommand{\zero}{\textbf{0}}
\newcommand{\bigO}{\mathcal{O}}
\newcommand{\bigOtil}{\widetilde{\mathcal{O}}}
\newcommand{\mydefn}{:=}

\begin{document}


\begin{center}

{\bf{\LARGE{On Efficient Optimal Transport: An Analysis of\\[.1cm]
Greedy and Accelerated Mirror Descent Algorithms}}}

\vspace*{.2in}
{\large{
\begin{tabular}{ccc}
Tianyi Lin$^{\star, \ddagger}$ & Nhat Ho$^{\star, \diamond}$ &  Michael I. Jordan$^{\diamond, \dagger}$ \\
 \end{tabular}
}}

\vspace*{.2in}

\begin{tabular}{c}
Department of Electrical Engineering and Computer Sciences$^\diamond$ \\
Department of Industrial Engineering and Operations Research$^\ddagger$ \\
Department of Statistics$^\dagger$ \\ 
University of California, Berkeley \\
\end{tabular}

\vspace*{.2in}

\today

\vspace*{.2in}

\begin{abstract} 
We provide theoretical analyses for two algorithms that solve the
regularized optimal transport (OT) problem between two discrete probability measures with at most $n$ atoms. We show that a greedy variant of the classical Sinkhorn algorithm, known as the \emph{Greenkhorn algorithm}, can be improved to $\bigOtil(n^2\varepsilon^{-2})$, improving on the best known complexity bound of $\bigOtil(n^2\varepsilon^{-3})$. Notably, this matches the best known complexity bound for the Sinkhorn algorithm and helps explain why the Greenkhorn algorithm can outperform the Sinkhorn algorithm in practice. Our proof technique, which is based on a primal-dual formulation and a novel upper bound for the dual solution, also leads to a new class of algorithms that we refer to as \emph{adaptive primal-dual accelerated mirror descent} (APDAMD) algorithms.  We prove that the complexity of these algorithms is $\bigOtil(n^2\sqrt{\delta}\varepsilon^{-1})$, where $\delta > 0$ refers to the inverse of the strong convexity module of Bregman divergence with respect to $\|\cdot\|_\infty$. This implies that the APDAMD algorithm is faster than the Sinkhorn and Greenkhorn algorithms in terms of $\varepsilon$.  Experimental results on synthetic and real datasets demonstrate the favorable performance of the Greenkhorn and APDAMD algorithms in practice.
\end{abstract}

\let\thefootnote\relax\footnotetext{$^\star$ Tianyi Lin and Nhat Ho contributed equally to this work.}
\end{center}

\section{Introduction}
Optimal transport---the problem of finding minimal cost couplings between pairs of probability measures---has a long history in mathematics and operations research~\citep{Villani-2003-Topic}. In recent years, it has been the inspiration for numerous applications in machine learning and statistics, including posterior contraction of parameter estimation in Bayesian nonparametrics models~\citep{Nguyen-2013-Convergence, Nguyen-2016-Borrowing}, scalable posterior sampling for large datasets~\citep{Srivastava-2015-WASP, Srivastava-2018-Scalable}, optimization models for clustering complex structured data~\citep{Ho-2018-Probabilistic}, deep generative models and domain adaptation in deep learning~\citep{Arjovsky-2017-Wasserstein, Gulrajani-2017-Improved, Courty-2017-Optimal, Tolstikhin-2018-Wasserstein}, and other applications~\citep{Rolet-2016-Fast, Peyre-2016-Averaging, Carriere-2017-Sliced, Lin-2018-Sparsemax}. These large-scale applications have placed significant new demands on the efficiency of algorithms for solving the optimal transport problem, and a new literature has begun to emerge to provide new algorithms and complexity analyses for optimal transport. 

The computation of the optimal-transport (OT) distance can be formulated as a linear programming problem and solved in principle by interior-point methods. 
The best known complexity bound in this formulation is $\bigOtil\left(n^{5/2}\right)$, achieved by an interior-point algorithm due to~\cite{Lee-2014-Path}. However, Lee and Sidford's method requires as a subroutine a practical implementation of the Laplacian linear system solver, which is not yet available in the literature.~\cite{Pele-2009-Fast} proposed an alternative, implementable interior-point method for OT with a complexity bound is $\bigOtil(n^3)$. Another prevalent approach for computing OT distance between two discrete probability measures involves regularizing the objective function by the entropy of the transportation plan. 
The resulting problem, referred to as \emph{entropic regularized OT} or simply \emph{regularized OT}~\citep{Cuturi-2013-Sinkhorn, Benamou-2015-Iterative}, is more readily solved than the original problem since the objective is strongly convex with respect to $\|\cdot\|_1$. The longstanding state-of-the-art algorithm for solving regularized OT is the Sinkhorn algorithm~\citep{Sinkhorn-1974-Diagonal, Knight-2008-Sinkhorn, Kalantari-2008-Complexity}. Inspired by the growing scope of applications for optimal transport, several new algorithms have emerged in recent years that have been shown empirically to have superior performance when compared to the Sinkhorn algorithm. An example includes the Greenkhorn algorithm~\citep{Altschuler-2017-Near, Chakrabarty-2018-better, Abid-2018-Greedy}, which is a greedy version of Sinkhorn algorithm. A variety of standard optimization algorithms have also been adapted to the OT setting, including accelerated gradient descent~\citep{Dvurechensky-2018-Computational}, quasi-Newton methods~\citep{Cuturi-2016-Smoothed, Blondel-2018-Smooth} and stochastic average gradient~\citep{Genevay-2016-Stochastic}. The theoretical analysis of these algorithms is still nascent.

Very recently,~\cite{Altschuler-2017-Near} have shown that both the Sinkhorn and Greenkhorn algorithm can achieve the near-linear time complexity for regularized OT. More specifically, they proved that the complexity bounds for both algorithms are $\bigOtil(n^2\varepsilon^{-3})$, where $n$ is the number of atoms (or equivalently dimension) of each probability measure and $\varepsilon$ is a desired tolerance. Later,~\cite{Dvurechensky-2018-Computational} improved the complexity bound for the Sinkhorn algorithm to $\bigOtil(n^2\varepsilon^{-2})$ and further proposed an adaptive primal-dual accelerated gradient descent (APDAGD), asserting a complexity bound of $\bigOtil(\min\{n^{9/4}\varepsilon^{-1}, n^2\varepsilon^{-2}\})$ for this algorithm. It is also possible to use a carefully designed Newton-type algorithm to solve the OT problem~\citep{Allen-2017-Much, Cohen-2017-Matrix}, by making use of a connection to matrix-scaling problems.~\cite{Blanchet-2018-Towards} and~\cite{Quanrud-2018-Approximating} provided a complexity bound of $\bigOtil(n^2\varepsilon^{-1})$ for Newton-type algorithms. However, these methods are complicated and efficient implementations are not yet available. Nonetheless, this complexity bound can be viewed as a theoretical benchmark for the algorithms that we consider in this paper.
\paragraph{Our Contributions.} The contribution is three-fold and can be summarized as follows: 
\begin{enumerate}
\item We improve the complexity bound for the Greenkhorn algorithm from $\bigOtil(n^2\varepsilon^{-3})$ to $\bigOtil(n^2\varepsilon^{-2})$, matching the best known complexity bound for the Sinkhorn algorithm. This analysis requires a new proof technique---the technique used in~\cite{Dvurechensky-2018-Computational} for analyzing the complexity of Sinkhorn algorithm is not applicable to the Greenkhorn algorithm. In particular, the Greenkhorn algorithm only updates a single row or column at a time and its per-iteration progress is accordingly more difficult to quantify than that of the Sinkhorn algorithm. In contrast, we employ a novel proof technique that makes use of a novel upper bound for the dual optimal solution in terms of $\|\cdot\|_\infty$. Our results also shed light on the better practical performance of the Greenkhorn algorithm compared the Sinkhorn algorithm. 
\item The smoothness of the dual regularized OT with respect to $\|\cdot\|_\infty$ allows us to formulate a novel \emph{adaptive primal-dual accelerated mirror descent} (APDAMD) algorithm for the OT problem.  Here the Bregman divergence is strongly convex and smooth with respect to $\|\cdot\|_\infty$. The resulting method involves an efficient line-search strategy~\citep{Nesterov-2006-Cubic} that is readily analyzed.  It can be adapted to problems even more general than regularized OT. It can also be viewed as a primal-dual extension of~\cite[Algorithm~1]{Tseng-2008-On} and a mirror descent extension of the APDAGD algorithm~\citep{Dvurechensky-2018-Computational}. We establish a complexity bound for the APDAMD algorithm of $\bigOtil(n^2\sqrt{\delta}\varepsilon^{-1})$, where $\delta > 0$ refers to the inverse of the strong convexity module of the Bregman divergence with respect to $\|\cdot\|_\infty$. In particular, $\delta = n$ if the Bregman divergence is simply chosen as $(1/2n)\|\cdot\|^2$. This implies that the APDAMD algorithm is faster than the Sinkhorn and Greenkhorn algorithms in terms of $\varepsilon$. Furthermore, we are able to provide a robustness result for the APDAMD algorithm (see Section~\ref{sec:experiments}). 
\item We show that there is a limitation in the derivation by~\cite{Dvurechensky-2018-Computational} of the complexity bound of $\bigOtil(\min\{n^{9/4}\varepsilon^{-1}, n^2\varepsilon^{-2}\})$. More specifically, the complexity bound in~\cite{Dvurechensky-2018-Computational} depends on a parameter which is not estimated explicitly. We provide a sharp lower bound for this parameter by a simple example (Proposition~\ref{proposition:tight_upper}), demonstrating that this parameter depends on $n$. Due to the dependence on $n$ of that parameter, we demonstrate that the complexity bound of APDAGD algorithm is indeed $\bigOtil (n^{2.5}/\varepsilon)$. This is slightly worse than the asserted complexity bound of $\bigOtil(\min\{n^{9/4}\varepsilon^{-1}, n^2\varepsilon^{-2}\})$ in terms of $n$.  Finally, our APDAMD algorithm potentially provides an improvement for the complexity of APDAGD algorithm  as its complexity bound is $\bigOtil(n^{2}\sqrt{\delta}/\varepsilon)$ and $\delta$ can be smaller than $n$.   
\end{enumerate}

\paragraph{Organization.} The remainder of the paper is organized as follows. In Section~\ref{sec:setup}, we provide the basic setup for regularized OT in primal and dual forms, respectively. Based on the dual form, we analyze the worst-case complexity of the Greenkhorn algorithm in Section~\ref{sec:greenkhorn}. In Section~\ref{sec:apdamd}, we propose the APDAMD algorithm for solving regularized OT and provide a theoretical complexity analysis. Section~\ref{sec:experiments} presents  experiments that illustrate the favorable performance of the Greenkhorn and APDAMD algorithms. Finally, we conclude in Section~\ref{sec:discussion}.

\paragraph{Notation.} We let $\Delta^n$ denote the probability simplex in $n - 1$ dimensions, for $n \geq 2$:  $\Delta^n = \{u = \left(u_1, \ldots, u_n\right) \in \Rspace^n: \sum_{i = 1}^{n} u_{i} = 1, \ u \geq 0\}$. Furthermore, $[n]$ stands for the set $\{1, 2, \ldots, n\}$ while $\Rspace^n_+$ stands for the set of all vectors in $\Rspace^n$ with nonnegative components for any $n \geq 1$. For a vector $x \in \Rspace^n$ and $1 \leq p \leq \infty$, we denote $\|x\|_p$ as the $\ell_p$-norm with $\|x\|$ as the $\ell_2$-norm and $\text{diag}(x)$ as the diagonal matrix with $x$ on the diagonal. For a matrix $A \in \Rspace^{n \times n}$, the notation $\text{vec}(A)$ stands for the vector in $\Rspace^{n^2}$ obtained from concatenating the rows and columns of $A$. The notations $r(A)$ and $c(A)$ stand for the row sum and column sum of $A$. $\one_n$ and $\zero_n$ stand for $n$-dimensional vectors with all of their components equal to $1$ and $0$ respectively. $\partial_x f$ refers to a partial gradient of $f$ with respect to $x$. Lastly, given the dimension $n$ and accuracy $\varepsilon$, the notation $a = \bigO\left(b(n,\varepsilon)\right)$ stands for the upper bound $a \leq C \cdot b(n, \varepsilon)$ where $C$ is independent of $n$ and $\varepsilon$. Similarly, the notation $a = \bigOtil(b(n, \varepsilon))$ indicates the previous inequality may depend on the logarithmic function of $n$ and $\varepsilon$, and where $C>0$. 

\section{The Optimal Transport Problem}\label{sec:setup}
In this section, we provide some background materials on the problem of computing the OT distance between two discrete probability measures with at most $n$ atoms. 

According to~\cite{Kantorovich-1942-Translocation}, the problem of approximating the optimal transportation distance is equivalent to solving the following linear programming problem:
\begin{eqnarray}\label{prob:OT}
\min\limits_{X \in \br^{n \times n}} \langle C, X\rangle \quad \st \ X\one_n = r, \ X^\top\one_n = c, \ X \geq 0,
\end{eqnarray}
where $X$ refers to the \textit{transportation plan}, $C = (C_{ij}) \in \br_+^{n \times n}$ stands for a cost matrix with non-negative components, and $r$ and $c$ refer to two known probability distributions in the simplex $\Delta^n$. The goal of the paper is to find a transportation plan $\hat{X} \in \br_{+}^{n \times n}$ satisfying marginal distribution constraints $\hat{X}\one_n = r$ and $\hat{X}^\top\one_n = c$ and the following bound
\begin{equation}\label{Criteria:Approximation}
\langle C, \hat{X}\rangle \ \leq \ \langle C, X^*\rangle + \varepsilon.
\end{equation}
Here $X^*$ is defined as an optimal transportation plan for the OT problem~\eqref{prob:OT}. For the sake of presentation, we respectively denote $\langle C, \hat{X}\rangle$ an \emph{$\varepsilon$-approximate transportation cost} and $\hat{X}$ an \emph{$\varepsilon$-approximate transportation plan} for the original problem.  

Since Eq.~\eqref{prob:OT} is a linear programming problem, we can solve it by the interior-point method; however, this method performs poorly on large-scale problems due to its high per-iteration computational cost. Seeking a formulation for OT distance that is more amenable to computationally efficient algorithms,~\cite{Cuturi-2013-Sinkhorn} proposed to solve an entropic regularized version of the OT problem in Eq.~\eqref{prob:OT}, which is given by
\begin{equation}\label{prob:regOT}
\min\limits_{X \in \br^{n \times n}} \ \langle C, X\rangle - \eta H(X), \quad \st \ X\one_n = r, \ X^\top\one_n = c.
\end{equation}
Here $\eta > 0$ and $H(X) = - \sum_{i, j = 1}^n X_{ij} \log(X_{ij})$ stand for the \textit{regularization parameter} and \textit{entropic regularization} respectively. It is clear that the entropic regularized OT problem in Eq.~\eqref{prob:regOT} is a convex optimization model with affine constraints. Thus, the dual form of entropic regularized OT problem is an unconstrained optimization model. To derive the dual, we begin with a Lagrangian:
\begin{equation*}
\LCal(X, \alpha, \beta) \ = \ \langle\alpha, r\rangle + \langle\beta, c\rangle + \langle C, X\rangle - \eta H(X) - \langle\alpha, X\one_n\rangle - \langle\beta, X^\top\one_n\rangle. 
\end{equation*} 
Note that the dual function is defined by $\varphi(\alpha, \beta) \mydefn \min_{X \in \Rspace^{n \times n}} \LCal(X, \alpha, \beta)$. Since $\LCal(\cdot, \alpha, \beta)$ is strictly convex and differentiable, we can solve for the minimum by setting $\partial_X \LCal(X, \alpha, \beta)$ to zero and obtain for any $i, j \in [n]$ that, 
\begin{equation*}
C_{ij} + \eta\left(1 + \log(X_{ij}) \right) - \alpha_i - \beta_j = 0 \ \Longrightarrow \ X_{ij} = e^{\frac{-C_{ij} + \alpha_i + \beta_j}{\eta} - 1}.
\end{equation*}
Therefore, we conclude that 
\begin{equation}\label{prob:dualregOT-old}
\varphi(\alpha, \beta) \ = \ \eta\left(\sum_{i,j=1}^n e^{- \frac{C_{ij} - \alpha_i - \beta_j}{\eta}-1}\right) - \langle\alpha, r\rangle - \langle \beta, c\rangle - 1. 
\end{equation}
To simplify the notation, we perform a change of variables, setting $u_i = \eta^{-1}\alpha_i - (1/2)\one_n$ and $v_j = \eta^{-1}\beta_j - (1/2)\one_n$. Thus, solving $\min_{\alpha, \beta \in \br^n} \varphi(\alpha, \beta)$ is equivalent to solving
\begin{equation*}
\min_{u, v \in \br^n} \ \eta \left(\sum_{i,j=1}^n e^{- \frac{C_{ij}}{\eta} + u_i + v_j} - \langle u, r\rangle - \langle v, c\rangle - 1\right). 
\end{equation*}
Letting $B(u, v) \mydefn  \diag(e^u) \ e^{-\frac{C}{\eta}} \ \diag(e^v)$, the dual problem $\max_{u, v \in \br^n} \varphi(u, v)$ reduces to
\begin{equation}\label{prob:dualregOT}
\min_{u, v \in \br^n} \ f(u, v) \ \mydefn \ \one_n^\top B(u, v)\one_n - \langle u, r\rangle - \langle v, c\rangle. 
\end{equation}
The problem in Eq.~\eqref{prob:dualregOT} is called the \textit{dual (entropic) regularized OT} problem.

\section{Greenkhorn}\label{sec:greenkhorn}
In this section, we present a complexity analysis for the Greenkhorn algorithm. In particular, we improve the existing best known complexity bound $\bigO (n^2\|C\|_\infty^3\log(n)\varepsilon^{-3})$~\citep{Altschuler-2017-Near} to $\bigO(n^2\|C\|_\infty^2 \log(n)\varepsilon^{-2})$, which matches the best known complexity bound for the Sinkhorn algorithm~\citep{Dvurechensky-2018-Computational}. 

To facilitate the discussion later, we present the Greenkhorn algorithm in pseudocode form in Algorithm~\ref{Algorithm:Greenkhorn} and its application to regularized OT in Algorithm~\ref{Algorithm:ApproxOT_Greenkhorn}. Here the function $\rho: \br_+ \times \br_+ \rightarrow \left[0, +\infty\right]$ is given by $\rho(a, b) = b - a + a\log(a/b)$, which measures the progress in the dual objective value between two consecutive iterates of Algorithm~\ref{Algorithm:Greenkhorn}. In addition, we observe that the optimality condition of the dual regularized OT problem in Eq.\eqref{prob:dualregOT} is $r(B(u^t, v^t)) - r = 0$ and $c(B(u^t, v^t)) - c = 0$. This leads to the quantity which measures the error of the $t$-th iterate of the Greenkhorn algorithm:
\begin{equation*}
E_t \ \mydefn \ \|r(B(u^t, v^t)) - r\|_1 + \|c(B(u^t, v^t)) - c\|_1.
\end{equation*}
Both the Sinkhorn and Greenkhorn procedures are coordinate descent algorithms for the dual regularized OT problem in Eq.~\eqref{prob:dualregOT}. Comparing with the Sinkhorn algorithm, which performs alternating updates of \textit{all} rows and columns, the Greenkhorn algorithm updates a \textit{single} row or column at each step. Thus, the Greenkhorn algorithm updates only $\bigO(n)$ entries per iteration, rather than $\bigO(n^2)$. After $r(e^{-\eta^{-1}C})$ and $c(e^{-\eta^{-1}C})$ are computed once at the beginning of the algorithm, the Greenkhorn algorithm can easily be implemented such that each iteration runs in only $\bigO(n)$ arithmetic operations~\citep{Altschuler-2017-Near}. 

Despite per-iteration cheap computational cost, it is difficult to quantify the per-iteration progress of the Greenkhorn algorithm and the proof techniques~\citep{Dvurechensky-2018-Computational} are not applicable to the Greenkhorn algorithm. We thus explore a different strategy which will be elaborated in the sequel.
\begin{algorithm}[!t]
\caption{\textsc{Greenkhorn}$(C, \eta, r, c, \varepsilon')$} \label{Algorithm:Greenkhorn}
\begin{algorithmic}
\STATE \textbf{Input:} $t = 0$ and $u^0 = v^0 = 0$. 
\WHILE{$E_t > \varepsilon'$}
\STATE Compute $I = \argmax_{1 \leq i \leq n} \rho(r_i, r_i(B(u^t, v^t)))$. 
\STATE Compute $J = \argmax_{1 \leq j \leq n} \rho(c_j, c_j(B(u^t, v^t)))$.
\IF{$\rho(r_i, r_i(B(u^t, v^t))) > \rho(c_j, c_j(B(u^t, v^t)))$}
\STATE $u_I^{t+1} = u_I^t + \log(r_I) - \log(r_I(B(u^t, v^t)))$. 
\ELSE
\STATE $v_J^{t+1} = v_J^t + \log(c_J) - \log(c_J(B(u^t, v^t)))$.
\ENDIF
\STATE Increment by $t = t + 1$. 
\ENDWHILE
\STATE \textbf{Output:} $B(u^t, v^t)$.  
\end{algorithmic}
\end{algorithm} 
\begin{algorithm}[!t]
\caption{Approximating OT by Algorithm~\ref{Algorithm:Greenkhorn}} \label{Algorithm:ApproxOT_Greenkhorn}
\begin{algorithmic}
\STATE \textbf{Input:} $\eta = \varepsilon/(4\log(n))$ and $\varepsilon'=\varepsilon/(8\|C\|_\infty)$. 
\STATE \textbf{Step 1:} Let $\tilde{r} \in \Delta_n$ and $\tilde{c} \in \Delta_n$ be defined by $(\tilde{r}, \tilde{c}) = (1 - \varepsilon'/8)(r, c) + (\varepsilon'/8n)(\one_n, \one_n)$. 
\STATE \textbf{Step 2:} Compute $\tilde{X} = \textsc{Greenkhorn}(C, \eta, \tilde{r}, \tilde{c}, \varepsilon'/2)$. 
\STATE \textbf{Step 3:} Round $\tilde{X}$ to $\hat{X}$ using~\citet[Algorithm~2]{Altschuler-2017-Near} such that $\hat{X}\one_n = r$ and $\hat{X}^\top\one_n = c$.
\STATE \textbf{Output:} $\hat{X}$.  
\end{algorithmic}
\end{algorithm} 
\subsection{Complexity analysis---bounding dual objective values}
Given the definition of $E_t$, we first prove the following lemma which yields an upper bound for the objective values of the iterates.
\begin{lemma}\label{Lemma:Descent}
Letting $\{(u^t, v^t)\}_{t \geq 0}$ be the iterates generated by Algorithm~\ref{Algorithm:Greenkhorn}, we have
\begin{equation} \label{eq:upper_func}
f(u^t, v^t) - f(u^*, v^*) \ \leq \ 2(\|u^*\|_\infty + \|v^*\|_\infty) E_t, 
\end{equation}
where $(u^*, v^*)$ is an optimal solution of the dual regularized OT problem in Eq.~\eqref{prob:dualregOT}. 
\end{lemma}
\begin{proof}
By the definition, we have
\begin{equation*}
f(u, v) \ = \ \one_n^\top B(u, v)\one_n - \langle u, r\rangle - \langle v, c\rangle \ = \ \sum_{i, j=1}^n e^{u_i + v_j - \frac{C_{ij}}{\eta}} - \sum_{i=1}^n r_i u_i - \sum_{j=1}^n c_j v_j. 
\end{equation*}
Since $\nabla_u f(u^t, v^t) = B(u^t, v^t)\one_n - r$ and $\nabla_v f(u^t, v^t) = B(u^t, v^t)^\top\one_n - c$, the quantity $E_t$ can be rewritten as $E_t = \|\nabla_u f(u^t, v^t)\|_1 + \|\nabla_v f(u^t, v^t)\|_1$. Using the fact that $f$ is convex and globally minimized at $(u^*, v^*)$, we have
\begin{equation*}
f(u^t, v^t) - f(u^*, v^*) \ \leq \ (u^t - u^*)^\top\nabla_u f(u^t, v^t) + (v^t - v^*)^\top\nabla_v f(u^t, v^t). 
\end{equation*}
Applying H$\ddot{\text{o}}$lder's inequality yields
\begin{eqnarray}\label{lemma-descent-inequality-main}
f(u^t, v^t) - f(u^*, v^*) & \leq & \|u^t - u^*\|_\infty\|\nabla_u f(u^t, v^t)\|_1 + \|v^t - v^*\|_\infty\|\nabla_v f(u^t, v^t)\|_1 \nonumber  \\
& = & \left(\|u^t - u^*\|_\infty + \|v^t - v^*\|_\infty\right) E_t. 
\end{eqnarray}
Thus, it suffices to show that 
\begin{equation*}
\|u^t - u^*\|_\infty + \|v^t - v^*\|_\infty \ \leq \ 2\left\|u^*\right\|_\infty + 2\left\|v^*\right\|_\infty. 
\end{equation*}
The next result is the key observation that makes our analysis work for the Greenkhorn algorithm. We use an induction argument to establish the following bound: 
\begin{equation}\label{lemma-descent-inequality-bound}
\max\{\|u^t - u^*\|_\infty, \ \|v^t - v^*\|_\infty\} \ \leq \ \max\{\|u^0 - u^*\|_\infty, \ \|v^0 - v^*\|_\infty\}. 
\end{equation}
It is easy to verify Eq.~\eqref{lemma-descent-inequality-bound} when $t=0$. Assuming that it holds true for $t = k_0 \geq 0$, we show that it also holds true for $t = k_0 + 1$. Without loss of generality, let $I$ be the index chosen at the $(k_0+1)$-th iteration. Then
\begin{align}
\|u^{k_0+1} - u^*\|_\infty & \leq \ \max \{\|u^{k_0} - u^*\|_\infty, |u_I^{k_0+1} - u_I^*|\}, \label{lemma-descent-inequality-bound-first} \\ 
\|v^{k_0+1} - v^*\|_\infty & = \ \|v^{k_0} - v^*\|_\infty. \label{lemma-descent-inequality-bound-second}
\end{align} 
Using the updating formula for $u_I^{k_0+1}$ and the optimality 
condition for $u_I^*$, we have
\begin{equation*}
e^{u_I^{k_0+1}} \ = \ \frac{r_I}{\sum_{j=1}^n e^{-\frac{C_{ij}}{\eta} + v_j^{k_{0}}}}, \qquad e^{u_I^*} \ = \ \frac{r_I}{\sum_{j=1}^n e^{-\frac{C_{ij}}{\eta} + v_j^*}}.
\end{equation*}
Putting these pieces together yields
\begin{equation}\label{lemma-descent-inequality-bound-third}
|u_I^{k_0+1} - u_I^*| \ = \ \left|\log\left(\frac{\sum_{j=1}^n e^{-C_{Ij}/\eta + v_j^{k_0}}}{\sum_{j=1}^n e^{-C_{Ij}/\eta + v_j^*}}\right)\right| \ \leq \ \|v^{k_0} - v^*\|_\infty,   
\end{equation}
where the inequality comes from the inequality: $(\sum_{i=1}^n a_i)/(\sum_{i=1}^n b_i) \leq \max_{1 \leq j \leq n} (a_i/b_i)$ for all $a_i, b_i > 0$. Combining Eq.~\eqref{lemma-descent-inequality-bound-first} and Eq.~\eqref{lemma-descent-inequality-bound-third} yields 
\begin{equation}\label{lemma-descent-inequality-bound-fourth}
\|u^{k_0+1} - u^*\|_\infty \ \leq \ \max\{\|u^{k_0} - u^*\|_\infty, \|v^{k_0} - v^*\|_\infty\}. 
\end{equation}
Combining Eq.~\eqref{lemma-descent-inequality-bound-second} and Eq.~\eqref{lemma-descent-inequality-bound-fourth} further implies the desired Eq.~\eqref{lemma-descent-inequality-bound}. This together with $u^0 = v^0 = 0$ implies
\begin{eqnarray}\label{lemma-descent-inequality-bound-main}
\|u^t - u^*\|_\infty + \|v^t - v^*\|_\infty & \leq & 2\left(\|u^0 - u^*\|_\infty + \|v^0 - v^*\|_\infty\right) \\
& = & 2\|u^*\|_\infty + 2\|v^*\|_\infty. \nonumber  
\end{eqnarray}
Putting Eq.~\eqref{lemma-descent-inequality-main} and Eq.~\eqref{lemma-descent-inequality-bound-main} together yields Eq.~\eqref{eq:upper_func}.
\end{proof}
Our second lemma provides an upper bound for the $\ell_\infty$-norm of one optimal solution $(u^*, v^*)$ of the dual regularized OT problem. This result is stronger than~\citet[Lemma~1]{Dvurechensky-2018-Computational} and generalize~\citet[Lemma~8]{Blanchet-2018-Towards}.   
\begin{lemma}\label{Lemma:Boundedness}
For the dual regularized OT problem in Eq.~\eqref{prob:dualregOT}, there exists an optimal solution $(u^*, v^*)$ such that 
\begin{equation}\label{Lemma-inequality-boundedness}
\|u^*\|_\infty \leq R, \qquad \|v^*\|_\infty \leq R, 
\end{equation}
where $R = \eta^{-1}\|C\|_\infty + \log(n) - 2\log(\min_{1 \leq i, j \leq n} \{r_i, c_j\})$ depends on $C$, $r$ and $c$. 
\end{lemma}
\begin{proof}
First, we claim that there exists an optimal solution pair $(u^*, v^*)$ such that 
\begin{equation}\label{eq:dif_sign}
\max_{1 \leq i \leq n} u_i^* \ \geq \ 0 \ \geq \ \min_{1 \leq i \leq n} u_i^*. 
\end{equation}
Indeed, since the function $f$ is convex with respect to $(u, v)$, the set of optima of problem in Eq.~\eqref{prob:dualregOT} is nonempty. Thus, we can choose an optimal solution $(\tilde{u}^*, \tilde{v}^*)$ where
\begin{equation*}
+\infty > \max_{1 \leq i \leq n} \tilde{u}_i^* \geq \min_{1 \leq i \leq n} \tilde{u}_i^* > - \infty, \qquad +\infty > \max_{1 \leq i \leq n} \tilde{v}_i^* \geq \min_{1 \leq i \leq n} \tilde{v}_i^* > - \infty.  
\end{equation*}
Given the optimal solution $(\tilde{u}^*, \tilde{v}^*)$, we let $(u^*, v^*)$ be
\begin{equation*}
u^* = \tilde{u}^* - \frac{(\max_{1 \leq i \leq n} u_i^* + \min_{1 \leq i \leq n} u_i^*)\one_n}{2}, \quad v^* = \tilde{v}^* + \frac{(\max_{1 \leq i \leq n} u_i^* + \min_{1 \leq i \leq n} u_i^*)\one_n}{2}. 
\end{equation*}
and observe that $(u^*, v^*)$ satisfies Eq.~\eqref{eq:dif_sign}. Then it suffices to show that $(u^*, v^*)$ is optimal; i.e., $f(u^*, v^*) = f(\tilde{u}^*, \tilde{v}^*)$. Since $\one_n^\top r = \one_n^\top c = 1$, we have $\langle u^*, r\rangle = \langle\tilde{u}^*, r\rangle$ and $\langle v^*, c\rangle = \langle\tilde{v}^*, c\rangle$. Therefore, we conclude that 
\begin{eqnarray*}
f(u^*, v^*) & = & \sum_{i, j=1}^n e^{- \eta^{-1}C_{ij} + u_i^* + v_j^*} - \langle u^*, r\rangle - \langle v^*, c\rangle \ = \ \sum_{i, j=1}^n e^{- \eta^{-1}C_{ij} + \tilde{u}_i^* + \tilde{v}_j^*} - \langle\tilde{u}^*, r\rangle - \langle\tilde{v}^*, c\rangle \\
& = & f(\tilde{u}^*, \tilde{v}^*).
\end{eqnarray*}
The next step is to establish the following bounds:
\begin{align}
\max_{1 \leq i \leq n} u_i^* - \min_{1 \leq i \leq n} u_i^* & \ \leq \ \frac{\|C\|_\infty}{\eta} - \log\left(\min_{1 \leq i, j \leq n} \{r_i, c_j\}\right), \label{Lemma-inequality-bound-u} \\
\max_{1 \leq i \leq n} v_i^* - \min_{1 \leq i \leq n} v_i^* & \ \leq \ \frac{\|C\|_\infty}{\eta} - \log\left(\min_{1 \leq i, j \leq n} \{r_i, c_j\}\right). \label{Lemma-inequality-bound-v}
\end{align}
Indeed, for each $1 \leq i \leq n$, we have
\begin{equation*}
e^{-\eta^{-1}\|C\|_\infty + u_i^*} 
\left(\sum_{j=1}^n e^{v_j^*}\right) \ \leq \ \sum_{j=1}^n e^{-\eta^{-1}C_{ij} + u_i^* + v_j^*} \ = \ [B(u^*, v^*)\one_n]_i \ = \ r_i \leq  1,
\end{equation*}
which implies $u_i^* \leq \eta^{-1}\|C\|_\infty - \log(\sum_{j=1}^n e^{v_j^*})$. Furthermore, we have
\begin{equation*}
e^{u_i^*}\left(\sum_{j=1}^n e^{v_j^*}\right) \ \geq \ \sum_{j=1}^n e^{-\eta^{-1}C_{ij} + u_i^* + v_j^*} \ = \ [B(u^*, v^*)\one_n]_i \ = \ r_i \ \geq \ \min_{1 \leq i, j \leq n} \{r_i, c_j\},
\end{equation*}
which implies $u_i^* \geq \log(\min_{1 \leq i, j \leq n} \{r_i, c_j\}) - \log(\sum_{j=1}^n e^{v_j^*})$. Putting these pieces together yields Eq.~\eqref{Lemma-inequality-bound-u}. Using the similar argument, we can prove Eq.~\eqref{Lemma-inequality-bound-v} holds true. 

Finally, we proceed to prove that Eq.~\eqref{Lemma-inequality-boundedness} holds true. First, we assume $\max_{1 \leq i \leq n} v_i^* \geq 0$ and $\max_{1 \leq i \leq n} u_i^* \geq 0 \geq \min_{1 \leq i \leq n} u_i^*$. Then the optimality condition implies  
\begin{equation*}
\max_{1 \leq i \leq n} u_i^* + \max_{1 \leq i \leq n} v_i^* \ \leq \ \log\left(\max_{1 \leq i, j \leq n} e^{\eta^{-1}C_{ij}}\right) \ = \ \frac{\|C\|_\infty}{\eta}. 
\end{equation*}
and $\sum_{i, j=1}^n e^{-\eta^{-1}C_{ij} + u_i^* + v_j^*} = 1$. This together with the assumptions $\max_{1 \leq i \leq n} u_i^* \geq 0$ and $\max_{1 \leq i \leq n} v_i^* \geq 0$ yields 
\begin{equation}\label{Lemma-inequality-u-v-lower-upper-bound}
0 \ \leq \ \max_{1 \leq i \leq n} u_i^* \ \leq \ \frac{\|C\|_\infty}{\eta}, \qquad 0 \ \leq \ \max_{1 \leq i \leq n} v_i^* \ \leq \ \frac{\|C\|_\infty}{\eta}. 
\end{equation}
Combining Eq.~\eqref{Lemma-inequality-u-v-lower-upper-bound} with Eq.~\eqref{Lemma-inequality-bound-u} and Eq.~\eqref{Lemma-inequality-bound-v} yields 
\begin{equation*}
\min_{1 \leq i \leq n} u_i^* \ \geq \ -\frac{\|C\|_\infty}{\eta} + \log\left(\min_{1 \leq i, j \leq n} \{r_i, c_j\}\right), \quad \min_{1 \leq i \leq n} v_i^* \ \geq \ -\frac{\|C\|_\infty}{\eta} + \log\left(\min_{1 \leq i, j \leq n} \{r_i, c_j\}\right).
\end{equation*}
which implies that Eq.~\eqref{Lemma-inequality-boundedness} holds true.  

We proceed to the alternative scenario, where $\max_{1 \leq i \leq n} v_i^* \leq 0$ and $\max_{1 \leq i \leq n} u_i^* \geq 0 \geq \min_{1 \leq i \leq n} u_i^*$. This together with~\eqref{Lemma-inequality-bound-u} yields
\begin{equation*}
-\frac{\|C\|_\infty}{\eta} + \log\left(\min_{1 \leq i, j \leq n} \{r_i, c_j\}\right) \ \leq \ \min_{1 \leq i \leq n} u_i^* \ \leq \ \max_{1 \leq i \leq n} u_i^* \ \leq \ \frac{\|C\|_\infty}{\eta} - \log\left(\min_{1 \leq i, j \leq n} \{r_i, c_j\}\right). 
\end{equation*}
Furthermore, we have
\begin{equation*}
\min_{1 \leq i \leq n} v_i^* \ \geq \ \log\left(\min_{1 \leq i, j \leq n} \{r_i, c_j\}\right) - \log\left(\sum_{i=1}^n e^{u_i^*}\right) \ \geq \ 2\log\left(\min_{1 \leq i, j \leq n} \{r_i, c_j\}\right) - \log(n) - \frac{\|C\|_\infty}{\eta} 
\end{equation*}
which implies that Eq.~\eqref{Lemma-inequality-boundedness} holds true. 
\end{proof}
Putting Lemma~\ref{Lemma:Descent} and~\ref{Lemma:Boundedness} together, we have the following straightforward consequence: 
\begin{corollary}\label{cor:bound_func}
Letting $\{(u^t, v^t)\}_{t \geq 0}$ be the iterates generated by Algorithm~\ref{Algorithm:Greenkhorn}, we have
\begin{equation}\label{Lemma-inequality-descent}
f(u^t, v^t) - f(u^*, v^*) \leq 4RE_t.
\end{equation}
\end{corollary}
\begin{remark}
The constant $R$ provides an upper bound in both~\citet{Dvurechensky-2018-Computational} and this paper, where the same notation is used. The values in the two papers are of the same order since $R$ in our paper only involves an additional term $\log(n) - \log(\min_{1\leq i, j \leq n} \{r_i, c_j\})$.
\end{remark}
\begin{remark}
We further comment on the proof techniques in this paper and~\citet{Dvurechensky-2018-Computational}. The proof for~\citet[Lemma~2]{Dvurechensky-2018-Computational} depends on taking full advantage of the shift property of the Sinkhorn algorithm; namely, either $B(\overline{u}^t, \overline{v}^t)\one_n = r$ or $B(\overline{u}^t, \overline{v}^t)^\top\one_n = c$ where $(\overline{u}^t, \overline{v}^t)$ stands for the iterate generated by the Sinkhorn algorithm. Unfortunately, the Greenkhorn algorithm does not enjoy such a shift property. We have thus proposed a different approach for bounding $f(u^t, v^t) - f(u^*, v^*)$, based on the $\ell_\infty$-norm of the optimal solution $(u^*, v^*)$ of the dual regularized OT problem.
\end{remark}

\subsection{Complexity analysis---bounding the number of iterations}
We proceed to provide an upper bound for the iteration number to achieve a desired tolerance $\varepsilon'$ in Algorithm~\ref{Algorithm:Greenkhorn}. First, we start with a lower bound for the difference of function values between two consecutive iterates of Algorithm~\ref{Algorithm:Greenkhorn}:
\begin{lemma} \label{lemma:func_consec}
Letting $\{(u^t, v^t)\}_{t \geq 0}$ be the iterates generated by Algorithm~\ref{Algorithm:Greenkhorn}, we have
\begin{equation*}
f(u^t, v^t) - f(u^{t+1}, v^{t+1}) \ \geq \ \frac{(E_t)^2}{28n}. 
\end{equation*}
\end{lemma}
\begin{proof}
We observe that 
\begin{eqnarray*}
f(u^t, v^t) - f(u^{t+1}, v^{t+1}) & \geq & \frac{1}{2n}\left(\rho(r, r(B(u^t, v^t)) + \rho(c, c(B(u^t, v^t))\right) \\
& & \hspace*{-8em} \geq \frac{1}{14n}\left(\|r - r(B(u^t, v^t))\|_1^2 + \|c - c(B(u^t, v^t))\|_1^2\right),
\end{eqnarray*}
where the first inequality comes from~\citet[Lemma~5]{Altschuler-2017-Near} and the fact that the row or column update is chosen in a greedy manner, and the second inequality comes from~\citet[Lemma~6]{Altschuler-2017-Near}. The definition of $E_t$ implies the desired result.
\end{proof}
We are now able to derive the iteration complexity of Algorithm~\ref{Algorithm:Greenkhorn}. 
\begin{theorem}\label{Theorem:Greenkhorn-Total-Complexity}
Letting $\{(u^t, v^t)\}_{t \geq 0}$ be the iterates generated by Algorithm~\ref{Algorithm:Greenkhorn}, the number of iterations required to satisfy $E_t \leq \varepsilon'$ is upper bounded by $t \leq 2 + 112nR/\varepsilon'$ where $R > 0$ is defined in Lemma~\ref{Lemma:Boundedness}.
\end{theorem}
\begin{proof}
Letting $\delta_t = f(u^t, v^t) - f(u^*, v^*)$, we derive from Corollary~\ref{cor:bound_func} and Lemma~\ref{lemma:func_consec} that
\begin{equation*}
\delta_t - \delta_{t+1} \ \geq \ \max\{\delta_t^2/(448nR^2), \ (\varepsilon')^2/(28n)\}, 
\end{equation*}
where $E_t \geq \varepsilon'$ as soon as the stopping criterion is not fulfilled. In the following step we apply a switching strategy introduced by~\citet{Dvurechensky-2018-Computational}. Given any $t \geq 1$, we have two estimates: 
\begin{itemize}
\item[(i)] Considering the process from the first iteration and the $t$-th iteration, we have
\begin{equation*}
\frac{\delta_{t+1}}{448nR^2} \ \leq \ \frac{1}{t + 448nR^2\delta_1^{-2}} \quad \Longrightarrow \quad t \ \leq \ 1 + \dfrac{448nR^2}{\delta_t} - \dfrac{448nR^2}{\delta_1}. 
\end{equation*}
\item[(ii)] Considering the process from the $(t+1)$-th iteration to the $(t+m)$-th iteration for any $m \geq 1$, we have
\begin{equation*}
\delta_{t+m} \ \leq \ \delta_t - \frac{(\varepsilon')^2 m}{28n} \quad \Longrightarrow \quad m \ \leq \ \dfrac{28n(\delta_t - \delta_{t+m})}{(\varepsilon')^2}. 
\end{equation*}
\end{itemize}
We then minimize the sum of two estimates by an optimal choice of $s \in \left(0, \delta_1\right]$: 
\begin{equation*}
t \ \leq \ \min_{0 < s \leq \delta_1} \left(2 + \frac{448nR^2}{s} - \frac{448nR^2}{\delta_1} + \frac{28ns}{(\varepsilon')^2}\right) \ = \ 
\left\{\begin{array}{ll}
2 + \dfrac{224nR}{\varepsilon'} - \dfrac{448nR^2}{\delta_1}, & \delta_1 \geq 4R\varepsilon', \\
2 + \dfrac{28n\delta_1}{(\varepsilon')^2}, & \delta_1 \leq 4R\varepsilon'. 
\end{array}
\right. 
\end{equation*}
This implies that $t \leq 2 + 112nR/\varepsilon'$ in both cases and completes the proof. 
\end{proof}
Equipped with the result of Theorem~\ref{Theorem:Greenkhorn-Total-Complexity} and the scheme of Algorithm~\ref{Algorithm:ApproxOT_Greenkhorn}, we are able to establish the following result for the complexity of Algorithm~\ref{Algorithm:ApproxOT_Greenkhorn}:
\begin{theorem}\label{Theorem:ApproxOT-Greenkhorn-Total-Complexity}
The Greenkhorn algorithm for approximating optimal transport 
(Algorithm~\ref{Algorithm:ApproxOT_Greenkhorn}) returns 
$\hat{X} \in \br^{n \times n}$ satisfying $\hat{X}\one_n = r$, $\hat{X}^\top\one_n = c$ and Eq.~\eqref{Criteria:Approximation} in
\begin{equation*}
\bigO\left(\frac{n^2\left\|C\right\|_\infty^2\log(n)}{\varepsilon^2}\right)
\end{equation*}
arithmetic operations. 
\end{theorem}
\begin{proof}
We follow the proof steps in~\citet[Theorem~1]{Altschuler-2017-Near} and obtain
\begin{eqnarray*}
\langle C, \hat{X}\rangle - \langle C, X^*\rangle & \leq & 2\eta\log(n) + 4(\|\tilde{X}\one_n - r\|_1 + \|\tilde{X}^\top\one_n - c\|_1)\|C\|_\infty \\
& \leq & \varepsilon/2 + 4(\|\tilde{X}\one_n - r\|_1 + \|\tilde{X}^\top\one_n - c\|_1)\|C\|_\infty, 
\end{eqnarray*}
where $\hat{X}$ is returned by Algorithm~\ref{Algorithm:ApproxOT_Greenkhorn}, $X^*$ is a solution to the optimal transport problem and $\tilde{X}$ is returned by Algorithm~\ref{Algorithm:Greenkhorn} with $\tilde{r}$, $\tilde{c}$ and $\varepsilon'/2$ in Step 3 of Algorithm~\ref{Algorithm:ApproxOT_Greenkhorn}. The last inequality in the above display holds true since $\eta = \varepsilon/(4\log(n))$. Furthermore, we have
\begin{eqnarray*}
\|\tilde{X}\one_n - r\|_1 + \|\tilde{X}^\top\one_n - c\|_1 & \leq & \|\tilde{X}\one_n - \tilde{r}\|_1 + \|\tilde{X}^\top\one_n - \tilde{c}\|_1 + \|r - \tilde{r}\|_1 + \|c - \tilde{c}\|_1 \\
& \leq & \varepsilon'/2 + \varepsilon'/4 + \varepsilon'/2 \ = \ \varepsilon'.
\end{eqnarray*}
We conclude that $\langle C, \hat{X}\rangle - \langle C, X^*\rangle \leq \varepsilon$ from that $\varepsilon' = \varepsilon/(8\|C\|_\infty)$. The remaining step is to analyze the complexity bound. It follows from Theorem~\ref{Theorem:Greenkhorn-Total-Complexity} and the definition of $\tilde{r}$ and $\tilde{c}$ in Algorithm~\ref{Algorithm:ApproxOT_Greenkhorn} that 
\begin{eqnarray*}
t \ \leq \ 2 + \frac{112nR}{\varepsilon'} & \leq & 2 + \frac{96n\|C\|_\infty}{\varepsilon}\biggr(\frac{\|C\|_\infty}{\eta} + \log(n) 
- 2\log\left(\min_{1 \leq i, j \leq n} \{r_i, c_j\}\right)\biggr) \\
& \leq & 2 + \frac{96n\|C\|_\infty}{\varepsilon}\left(\frac{4\|C\|_\infty\log(n)}{\varepsilon} + \log(n) - 2\log\left(\frac{\varepsilon}{64n\|C\|_\infty}\right)\right) \\
& = & \bigO\left(\frac{n\|C\|_\infty^2\log(n)}{\varepsilon^2}\right). 
\end{eqnarray*}
The total iteration complexity in Step 3 of Algorithm~\ref{Algorithm:ApproxOT_Greenkhorn} is bounded by $\bigO(n\|C\|_\infty^2\log(n)\varepsilon^{-2})$. Each iteration of Algorithm~\ref{Algorithm:Greenkhorn} requires $\bigO(n)$ arithmetic operations. Thus, the total amount of arithmetic operations is $\bigO(n^2\|C\|_\infty^2\log(n)\varepsilon^{-2})$. Furthermore, $\tilde{r}$ and $\tilde{c}$ in Step 2 of Algorithm~\ref{Algorithm:ApproxOT_Greenkhorn} can be found in $\bigO(n)$ arithmetic operations and~\citet[Algorithm~2]{Altschuler-2017-Near} requires $O(n^2)$ arithmetic operations. Therefore, we conclude that the total number of arithmetic operations is $\bigO(n^2\|C\|_\infty^2\log(n)\varepsilon^{-2})$. 
\end{proof}
The result of Theorem~\ref{Theorem:ApproxOT-Greenkhorn-Total-Complexity} 
improves the best known complexity bound $\bigOtil(n^2\epsilon^{-3})$ of the Greenkhorn algorithm~\citep{Altschuler-2017-Near, Abid-2018-Greedy}, and further matches the best known complexity bound of the Sinkhorn algorithm~\citep{Dvurechensky-2018-Computational}. This sheds light on the superior performance of the Greenkhorn algorithm in practice. 

\section{Adaptive Primal-Dual Accelerated Mirror Descent}\label{sec:apdamd}
In this section, we propose an adaptive primal-dual accelerated mirror descent (APDAMD) algorithm for solving the regularized OT problem in Eq.~\eqref{prob:regOT}. The APDAMD algorithm and its application to regularized OT are presented in Algorithm~\ref{Algorithm:APDAMD} and~\ref{Algorithm:ApproxOT_APDAMD}. We further show that the complexity bound of APDAMD is $\bigO(n^2\sqrt{\delta}\|C\|_\infty \log(n)\varepsilon^{-1})$, where $\delta>0$ depends on the mirror mapping $\phi$ in Algorithm~\ref{Algorithm:APDAMD}. 
\begin{algorithm}[!t]
\caption{\textsc{Apdamd}$(\varphi, A, b, \varepsilon')$} \label{Algorithm:APDAMD}
\begin{algorithmic}
\STATE \textbf{Input:} $t = 0$. 
\STATE \textbf{Initialization:} $\bar{\alpha}^0 = \alpha^0 = 0$,  $z^0 = \mu^0 = \lambda^0$ and $L^0 = 1$. 
\REPEAT
\STATE Set $M^t = L^t/2$. 
\REPEAT
\STATE Set $M^t = 2M^t$. 
\STATE Compute the stepsize: $\alpha^{t+1} = (1 + \sqrt{1 + 4\delta M^t\bar{\alpha}^t})/(2\delta M^t)$. 
\STATE Compute the average coefficient: $\bar{\alpha}^{t+1} = \bar{\alpha}^t + \alpha^{t+1}$. 
\STATE Compute the first average step: $\mu^{t+1} = (\alpha^{t+1}z^t + \bar{\alpha}^t\lambda^t)/\bar{\alpha}^{t+1}$. 
\STATE Compute the mirror descent: $z^{t+1} = \argmin\limits_{z \in \br^n} \{\langle\nabla \varphi(\mu^{t+1}), z - \mu^{t+1}\rangle + B_\phi(z, z^t)/\alpha^{t+1}\}$. 
\STATE Compute the second average step: $\lambda^{t+1} = (\alpha^{t+1}z^{t+1} + \bar{\alpha}^t\lambda^t)/\bar{\alpha}^{t+1}$. 
\UNTIL{$\varphi(\lambda^{t+1}) - \varphi(\mu^{t+1}) - \langle\nabla \varphi(\mu^{t+1}), \lambda^{t+1} - \mu^{t+1}\rangle
\leq (M^t/2)\|\lambda^{t+1} - \mu^{t+1}\|_\infty^2$.}
\STATE Compute the main average step: $x^{t+1} = (\alpha^{t+1} x(\mu^{t+1}) + \bar{\alpha}^t x^t)/\bar{\alpha}^{t+1}$.  
\STATE Set $L^{t+1} = M^t/2$. 
\STATE Set $t = t + 1$. 
\UNTIL{$\|Ax^t - b\|_1 \leq \varepsilon'$.}
\STATE \textbf{Output:} $X^t$ where $x^t = \text{vec}(X^t)$.  
\end{algorithmic}
\end{algorithm}  
\begin{algorithm}[!t]
\caption{Approximating OT by Algorithm~\ref{Algorithm:APDAMD}} \label{Algorithm:ApproxOT_APDAMD}
\begin{algorithmic}
\STATE \textbf{Input:} $\eta = \varepsilon/(4\log(n))$ and $\varepsilon' = \varepsilon/(8\|C\|_\infty)$. 
\STATE \textbf{Step 1:} Let $\tilde{r} \in \Delta_n$ and $\tilde{c} \in \Delta_n$ be defined by $(\tilde{r}, \tilde{c}) = (1 - \varepsilon'/8)(r, c) + (\varepsilon'/8n)(\one_n, \one_n)$. 
\STATE \textbf{Step 2:} Let $A \in \br^{2n \times n^2}$ and $b \in \br^{2n}$ be defined by $A\text{vec}(X) = \begin{pmatrix} X\one_n \\ X^\top\one_n \end{pmatrix}$ and $b = \begin{pmatrix} \tilde{r} \\ \tilde{c} \end{pmatrix}$. 
\STATE \textbf{Step 3:} Compute $\tilde{X} = \textsc{Apdamd}(\varphi, A, b, \varepsilon'/2)$ where $\varphi$ is defined by Eq.~\eqref{prob:dualregOT-APDAMD}. 
\STATE \textbf{Step 4:} Round $\tilde{X}$ to $\hat{X}$ using~\citet[Algorithm~2]{Altschuler-2017-Near} such that $\hat{X}\one_n = r$ and $\hat{X}^\top\one_n = c$. 
\STATE \textbf{Output:} $\hat{X}$.  
\end{algorithmic}
\end{algorithm} 
\subsection{General setup}
We consider the following generalization of the regularized OT problem:
\begin{equation}\label{prob:generalOT}
\min_{x \in \br^n} f(x), \quad \st \ Ax = b, 
\end{equation}
where $A \in \br^{n \times n}$ is a matrix and $b \in \br^n$ and $f$ is assumed to be strongly convex with respect to $\ell_1$-norm: $f(x') - f(x) - \langle\nabla f(x), x' - x\rangle \geq (\eta/2)\|x' - x\|_1^2$. By abuse of notation, we use the same symbol here as Eq.~\eqref{prob:dualregOT-old} and obtain that the dual problem is as follows:
\begin{equation}\label{prob:generalOT-dual-objective}
\min_{\lambda \in \br^n} \varphi(\lambda) \ \mydefn \ \{\langle\lambda, b\rangle + \max_{x \in \br^n} \{-f(x) - \langle A^\top\lambda, x\rangle\}\}, 
\end{equation}
and $\nabla \varphi(\lambda) = b - Ax(\lambda)$ where $x(\lambda) = \argmax_{x \in \br^n} \{-f(x) - \langle A^\top\lambda, x\rangle\}$. To analyze the complexity bound of the APDAMD algorithm, we start with the following result that establishes the smoothness of the dual objective function $\varphi$ with respect to $\ell_\infty$-norm.
\begin{lemma}\label{Lemma:dualOT-smoothness}
The dual objective $\varphi$ is smooth with respect to $\ell_\infty$-norm:
\begin{equation*}
\varphi(\lambda_1) - \varphi(\lambda_2) - \langle\nabla \varphi(\lambda_2), \lambda_1 - \lambda_2\rangle \ \leq \ \frac{\|A\|_{1\rightarrow 1}^2}{2\eta}\left\|\lambda_1 - 
\lambda_2\right\|_\infty^2. 
\end{equation*}
\end{lemma}
\begin{proof}
First, we show that 
\begin{equation}\label{lemma-dualOT-gradient-bound}
\|\nabla\varphi(\lambda_1) - \nabla\varphi(\lambda_2)\|_1 \ \leq \ \frac{\|A\|_{1\rightarrow 1}^2}{\eta} \|\lambda_1 - \lambda_2\|_\infty. 
\end{equation}
Using the definition of $\nabla\varphi(\lambda)$ and $\|\cdot\|_{1 \rightarrow 1}$, we have
\begin{equation}\label{lemma-dualOT-gradient-bound-first}
\|\nabla\varphi(\lambda_1) - \nabla\varphi(\lambda_2)\|_1 \ = \ \|Ax(\lambda_1) - Ax(\lambda_2)\|_1 \ \leq \ \|A\|_{1\rightarrow 1}\|x(\lambda_1) - x(\lambda_2)\|_1. 
\end{equation}
Furthermore, it follows from the strong convexity of $f$ that 
\begin{eqnarray*}
\eta\|x(\lambda_1) - x(\lambda_2)\|_1^2 & \leq & \langle \nabla f(x(\lambda_1)) - \nabla f(x(\lambda_2)), x(\lambda_1) - x(\lambda_2)\rangle \ = \ \langle A^\top\lambda_2 - A^\top\lambda_1, x(\lambda_1) - x(\lambda_2)\rangle \\
& \leq & \|A\|_{1\rightarrow 1}\|x(\lambda_1) - x(\lambda_2)\|_1\|\lambda_1 - \lambda_2\|_\infty,
\end{eqnarray*}
which implies 
\begin{equation}\label{lemma-dualOT-gradient-bound-second}
\|x(\lambda_1) - x(\lambda_2)\|_1 \ \leq \ \frac{\|A\|_{1\rightarrow 1}}{\eta}\|\lambda_1 - \lambda_2\|_\infty. 
\end{equation}
Putting Eq.~\eqref{lemma-dualOT-gradient-bound-first} and Eq.~\eqref{lemma-dualOT-gradient-bound-second} yields Eq.~\eqref{lemma-dualOT-gradient-bound}. Then we have
\begin{eqnarray*}
\varphi(\lambda_1) - \varphi(\lambda_2) - \langle\nabla\varphi(\lambda_2), \lambda_1 - \lambda_2\rangle & = & \int_0^1 \langle\nabla\varphi(t\lambda_1+(1-t)\lambda_2) - \nabla\varphi(\lambda_2), \lambda_1 - \lambda_2\rangle \; dt \\
& & \hspace*{-6em} \leq \ \left(\int_0^1 \|\nabla\varphi(t\lambda_1+(1-t)\lambda_2) - \nabla\varphi(\lambda_2)\|_1 \; dt\right) \|\lambda_1 - \lambda_2\|_\infty. 
\end{eqnarray*}
Using Eq.~\eqref{lemma-dualOT-gradient-bound}, we have
\begin{equation*}
\varphi(\lambda_1) - \varphi(\lambda_2) - \langle\nabla\varphi(\lambda_2), \lambda_1 - \lambda_2\rangle \leq \left(\int_0^1 t \ dt\right)\frac{\|A\|_{1\rightarrow 1}^2}{\eta}\|\lambda_1 - \lambda_2\|_\infty^2 = \frac{\|A\|_{1\rightarrow 1}^2}{2\eta}\|\lambda_1 - \lambda_2\|_\infty^2. 
\end{equation*}
This completes the proof. 
\end{proof}
\begin{remark}
It is important to note that the objective function in Eq.~\eqref{prob:dualregOT-old} is the sum of exponents and its gradient is not Lipschitz yet. However, this does not contradict Lemma~\ref{Lemma:dualOT-smoothness}. Indeed, the entropy regularization is only strongly convex on the probability simplex while we have not considered the corresponding linear constraint $\one_n^\top X\one_n = 1$ for deriving the dual function before. To derive the smooth dual function, we shall consider the following minimization problem,  
\begin{equation*}
\min_{\one_n^\top X\one_n = 1} \LCal(X, \alpha, \beta). 
\end{equation*}
Using the same argument as before, we conclude that the dual function $\varphi$ is defined by 
\begin{equation}\label{prob:dualregOT-APDAMD}
\varphi(\alpha, \beta) \ = \ \eta\log\left(\sum_{i,j=1}^n e^{- \frac{C_{ij} - \alpha_i - \beta_j}{\eta}-1}\right) - \langle\alpha, r\rangle - \langle \beta, c\rangle - 1. 
\end{equation}
This function has the form of the logarithm of sum of exponents and hence has Lipschitz continuous gradient. Our APDAMD algorithm is developed for solving the regularized OT problem in Eq.~\eqref{prob:regOT} by using the function $\varphi$ in Eq.~\eqref{prob:dualregOT-APDAMD}. 
\end{remark}
To facilitate the ensuing discussion, we assume that the dual problem in Eq.~\eqref{prob:generalOT} has a solution $\lambda^*\in\br^n$. The Bregman divergence $B_\phi: \br^n \times \br^n \rightarrow \left[0, +\infty\right]$ is defined by 
\begin{equation*}
B_\phi(z, z') \ \mydefn \ \phi(z) - \phi(z') - \left\langle \phi(z'), z - z'\right\rangle, \quad \forall z, z' \in \br^n. 
\end{equation*}
The mirror mapping $\phi$ is $(1/\delta)$-strongly convex and 1-smooth on $\br^n$ with respect to $\ell_\infty$-norm. That is to say,  
\begin{equation}\label{eq:inf_const}
\frac{1}{2\delta}\|z - z'\|_\infty^2 \leq \phi(z) - \phi(z') - \langle\nabla\phi(z'), z - z'\rangle \leq \frac{1}{2}\|z - z'\|_\infty^2. 
\end{equation} 
For example, we can choose $\phi(\lambda) = \frac{1}{2n}\|\lambda\|^2$ and $B_\phi(\lambda', \lambda) = \frac{1}{2n}\|\lambda'-\lambda\|^2$ in the APDAMD algorithm where $\delta=n$. As such, $\delta > 0$ is a function of $n$ in general and it will appear in the complexity bound of the APDAMD algorithm for approximating the OT problem (cf. Theorem~\ref{Theorem:ApproxOT-APDAMD-Total-Complexity}). It is worth noting that our algorithm uses a regularizer that acts only in the dual and our complexity bound is the best existing one among this group of algorithms~\citep{Dvurechensky-2018-Computational}. A very recent work of~\citet{Jambulapati-2019-Direct} showed that the complexity bound can be improved to $\bigOtil(n^2\varepsilon^{-1})$ using a more advanced area-convex mirror mapping. 

\subsection{Properties of the APDAMD algorithm}
In this section, we present several important properties of Algorithm~\ref{Algorithm:APDAMD} that can be used later for regularized OT problems. First, we prove the following result regarding the number of line search iterations in Algorithm~\ref{Algorithm:APDAMD}: 
\begin{lemma}\label{Lemma:Line-search-iteration}
The number of line search iterations in Algorithm~\ref{Algorithm:APDAMD} is finite. Furthermore, the total number of gradient oracle calls after the $t$-th iteration is bounded as
\begin{equation}\label{inequality-line-search-iteration-main}
N_t \ \leq \ 4t + 4 + \frac{2\log(\|A\|_{1\rightarrow 1}^2/(2\eta)) - 2\log(L^0)}{\log 2}. 
\end{equation}
\end{lemma}
\begin{proof}
First, we observe that multiplying $M^t$ by two will not stop until the line search stopping criterion is satisfied. Therefore, we have $M^t \geq (\|A\|_{1\rightarrow 1}^2)/(2\eta)$. Using Lemma~\ref{Lemma:dualOT-smoothness}, we obtain that the number of line search iterations in the line search strategy is finite. Letting $i_j$ denote the total number of multiplication at the $j$-th iteration, we have 
\begin{equation*}
i_0 \ \leq \ 1 + \frac{\log(M^0/L^0)}{\log 2}, \qquad i_j \ \leq \ 2 + \frac{\log(M^j/M^{j-1})}{\log 2}. 
\end{equation*}
We claim that $M^j \leq \eta^{-1}\|A\|_{1\rightarrow 1}^2$ holds true. Otherwise, the line search stopping criterion is satisfied with $M^j/2$ since $M^j/2 \geq \eta^{-1}\|A\|_{1\rightarrow 1}^2$. Therefore, the total number of line search is bounded by
\begin{eqnarray*}
\sum_{j=0}^t i_j & \leq & 1 + \frac{\log(M^0/L^0)}{\log 2} + \sum_{j=1}^t \left(2 + \frac{\log(M^j/M^{j-1})}{\log 2}\right) \ \leq \ 2t + 1 + \frac{\log(M^t) - \log(L^0)}{\log 2} \\
& \leq & 2t + 1 + \frac{\log(\|A\|_{1\rightarrow 1}^2/2\eta) - \log(L^0)}{\log 2}. 
\end{eqnarray*}
The desired result follows since each line search contains two gradient oracle calls.
\end{proof}
The next lemma presents a property of the dual objective function in Algorithm~\ref{Algorithm:APDAMD}.
\begin{lemma}\label{Lemma:estimate-sequence}
For each iteration $t$ of Algorithm~\ref{Algorithm:APDAMD} and any $z \in \br^n$, we have
\begin{equation}\label{inequality-estimate-sequence}
\bar{\alpha}^t\varphi(\lambda^t) \ \leq \ \sum_{j=0}^t (\alpha^j (\varphi(\mu^j) + \langle\nabla\varphi(\mu^j), z - \mu^j\rangle)) + \|z\|_\infty^2.  
\end{equation}
\end{lemma}
\begin{proof}
First, we claim that it holds for any $z \in \br^n$:  
\begin{equation} \label{Inequality:main-first}
\alpha^{t+1}\langle\nabla\varphi(\mu^{t+1}), z^t - z\rangle \ \leq \ \bar{\alpha}^{t+1}(\varphi(\mu^{t+1}) - \varphi(\lambda^{t+1})) + B_\phi(z, z^t) - B_\phi(z, z^{t+1}). 
\end{equation}
Indeed, the optimality condition in mirror descent implies that, for any $z \in \br^n$, we have
\begin{equation}\label{Inequality:opt-MD}
\langle\nabla \varphi(\mu^{t+1}) + (\nabla \phi(z^{t+1}) - \nabla \phi(z^t))/\alpha^{t+1}, z - z^{t+1}\rangle \ \geq \ 0. 
\end{equation}
Recalling that the definition of Bregman divergence implies that $B_\phi(z, z^t) - B_\phi(z, z^{t+1}) - B_\phi(z^{t+1}, z^t) = \langle\nabla \phi(z^{t+1}) - \nabla \phi(z^t), z - z^{t+1}\rangle$, we have
\begin{eqnarray}\label{Inequality:MD-main}
& & \alpha^{t+1}\langle\nabla\varphi(\mu^{t+1}), z^t - z\rangle \\
& = & \alpha^{t+1}\langle\nabla\varphi(\mu^{t+1}), z^t - z^{t+1}\rangle + \alpha^{t+1}\langle\nabla\varphi(\mu^{t+1}), z^{t+1} - z\rangle \nonumber \\
& \overset{\text{Eq.~\eqref{Inequality:opt-MD}}}{\leq} & \alpha^{t+1}\langle\nabla \varphi(\mu^{t+1}), z^t - z^{t+1}\rangle + \langle\nabla\phi(z^{t+1}) - \nabla \phi(z^t), z - z^{t+1}\rangle \nonumber \\
& = & \alpha^{t+1}\langle\nabla\varphi(\mu^{t+1}), z^t - z^{t+1}\rangle + B_\phi(z, z^t) - B_\phi(z, z^{t+1}) - B_\phi(z^{t+1}, z^t) \nonumber \\
& \overset{\text{Eq.~\eqref{eq:inf_const}}}{\leq} & \alpha^{t+1}\langle\nabla\varphi(\mu^{t+1}), z^t - z^{t+1}\rangle + B_\phi(z, z^t) - B_\phi(z, z^{t+1}) - \|z^{t+1} - z^t\|_\infty^2/(2\delta), \nonumber
\end{eqnarray}
Furthermore, the update formulas of $\mu^{t+1}$, $\lambda^{t+1}$, $\alpha^{t+1}$ and $\bar{\alpha}^{t+1}$ imply that  
\begin{equation*}
\lambda^{t+1} - \mu^{t+1} = (\alpha^{t+1}/\bar{\alpha}^{t+1})(z^{t+1} - z^t), \qquad \delta M^t(\alpha^{t+1})^2 = \bar{\alpha}^{t+1}.
\end{equation*}
Putting these pieces together yields that 
\begin{equation*}
\alpha^{t+1}\langle\nabla\varphi(\mu^{t+1}), z^t - z^{t+1}\rangle \ = \ \bar{\alpha}^{t+1}\langle\nabla \varphi(\mu^{t+1}), \mu^{t+1} - \lambda^{t+1}\rangle.
\end{equation*}
and 
\begin{equation*}
\|z^{t+1} - z^t\|_\infty^2 \ = \ (\bar{\alpha}^{t+1}/\alpha^{t+1})^2\|\mu^{t+1} - \lambda^{t+1}\|_\infty^2 \ = \ \delta M^t\bar{\alpha}^{t+1}\|\mu^{t+1} - \lambda^{t+1}\|_\infty^2. 
\end{equation*}
Putting these pieces together with Eq.~\eqref{Inequality:MD-main} yields that  
\begin{eqnarray*}
& & \alpha^{t+1}\langle\nabla\varphi(\mu^{t+1}), z^t - z\rangle \\
& \leq & \bar{\alpha}^{t+1}\langle\nabla\varphi(\mu^{t+1}), \mu^{t+1} - \lambda^{t+1}\rangle + B_\phi(z, z^t) - B_\phi(z, z^{t+1}) - (\bar{\alpha}^{t+1} M^t/2)\|\mu^{t+1} - \lambda^{t+1}\|_\infty^2 \\
& = & \bar{\alpha}^{t+1}(\langle\nabla\varphi(\mu^{t+1}), \mu^{t+1} - \lambda^{t+1}\rangle - (M^t/2)\|\mu^{t+1} - \lambda^{t+1}\|_\infty^2) + B_\phi(z, z^t) - B_\phi(z, z^{t+1}) \\
& \leq & \bar{\alpha}^{t+1}(\varphi(\mu^{t+1}) - \varphi(\lambda^{t+1})) + B_\phi(z, z^t) - B_\phi(z, z^{t+1}), \nonumber
\end{eqnarray*}
where the last inequality comes from the stopping criterion in the 
line search. Therefore, we conclude that the desired Eq.~\eqref{Inequality:main-first} holds true. 

The next step is to bound the iterative objective gap, i.e., for $z \in \br^n$, 
\begin{equation}\label{Inequality:main-second}\small
\bar{\alpha}^{t+1}\varphi(\lambda^{t+1}) - \bar{\alpha}^t\varphi(\lambda^t) \leq \alpha^{t+1}(\varphi(\mu^{t+1}) + \langle\nabla\varphi(\mu^{t+1}), z - \mu^{t+1}\rangle) + B_\phi(z, z^t) - B_\phi(z, z^{t+1}). 
\end{equation} 
Combining $\bar{\alpha}^{t+1} = \bar{\alpha}^t + \alpha^{t+1}$ and the update formula of $\mu^{t+1}$ yields that 
\begin{eqnarray*}
\alpha^{t+1}(\mu^{t+1} - z^t) & = & (\bar{\alpha}^{t+1} - \bar{\alpha}^t)\mu^{t+1} - \alpha^{t+1} z^t \ = \ \alpha^{t+1}z^t + \bar{\alpha}^t\lambda^t - \bar{\alpha}^t\mu^{t+1} - \alpha^{t+1} z^t \\
& = & \bar{\alpha}^t(\lambda^t - \mu^{t+1}). \nonumber
\end{eqnarray*}
This together with the convexity of $\varphi$ implies that 
\begin{eqnarray*}
\alpha^{t+1}\langle\nabla \varphi(\mu^{t+1}), \mu^{t+1} - z\rangle
& = & \alpha^{t+1}\langle \nabla \varphi(\mu^{t+1}), \mu^{t+1} - z^t\rangle + \alpha^{t+1}\langle \nabla \varphi(\mu^{t+1}), z^t - z\rangle \\
& = & \bar{\alpha}^t\langle \nabla\varphi(\mu^{t+1}), \lambda^t - \mu^{t+1}\rangle + \alpha^{t+1}\langle \nabla\varphi(\mu^{t+1}), z^t - z\rangle \\
& \leq & \bar{\alpha}^t(\varphi(\lambda^t) - \varphi(\mu^{t+1})) + \alpha^{t+1}\langle \nabla\varphi(\mu^{t+1}), z^t - z\rangle. 
\end{eqnarray*}
Furthermore, we derive from Eq.~\eqref{Inequality:main-first} and $\bar{\alpha}^{t+1} = \bar{\alpha}^t + \alpha^{t+1}$ that 
\begin{eqnarray*}
& & \bar{\alpha}^t(\varphi(\lambda^t) - \varphi(\mu^{t+1})) + \alpha^{t+1}\langle \nabla\varphi(\mu^{t+1}), z^t - z\rangle \\
& \overset{~\eqref{Inequality:main-first}}{\leq} & \bar{\alpha}^t(\varphi(\lambda^t) - \varphi(\mu^{t+1})) + \bar{\alpha}^{t+1}(\varphi(\mu^{t+1}) - \varphi(\lambda^{t+1})) + B_\phi(z, z^t) - B_\phi(z, z^{t+1}) \\
& = & \bar{\alpha}^t\varphi(\lambda^t) - \bar{\alpha}^{t+1}\varphi(\lambda^{t+1}) + \alpha^{t+1}\varphi(\mu^{t+1}) + B_\phi(z, z^t) - B_\phi(z, z^{t+1}). 
\end{eqnarray*}
Putting these pieces together yields that Eq.~\eqref{Inequality:main-second} holds true. Summing up Eq.~\eqref{Inequality:main-second} over $t = 0, 1, \ldots, N-1$ and using $B_\phi(z, z^N) \geq 0$, we have
\begin{equation*}
\bar{\alpha}^N\varphi(\lambda^N) - \bar{\alpha}^0\varphi(\lambda^0) \ \leq \  \sum_{t=0}^{N-1} (\alpha^{t+1}(\varphi(\mu^{t+1}) + \langle\nabla\varphi(\mu^{t+1}), z - \mu^{t+1}\rangle)) + B_\phi(z, z^0). 
\end{equation*}
Since $\alpha^0 = \bar{\alpha}^0 = 0$ and $\phi$ is 1-smooth with respect to $\ell_\infty$-norm, we conclude that 
\begin{eqnarray*}
\bar{\alpha}^N\varphi(\lambda^N) & \leq & \sum_{t=0}^N (\alpha^t (\varphi(\mu^t) + \langle\nabla\varphi(\mu^t), z - \mu^t\rangle)) + B_\phi(z, z^0) \\
& \leq & \sum_{t=0}^N (\alpha^t (\varphi(\mu^t) + \langle \nabla\varphi(\mu^t), z - \mu^t\rangle)) + \|z-z^0\|_\infty^2 \\
& \overset{z^0=0}{=} & \sum_{t=0}^N (\alpha^t(\varphi(\mu^t) + \langle \nabla\varphi(\mu^t), z - \mu^t\rangle)) + \|z\|_\infty^2.
\end{eqnarray*}
which implies that Eq.~\eqref{inequality-estimate-sequence} holds true.
\end{proof}
The final lemma provides us with a key lower bound for the accumulating parameter.
\begin{lemma}\label{Lemma:objective-coefficient}
For each iteration $t$ of Algorithm~\ref{Algorithm:APDAMD}, we have $\bar{\alpha}^t \geq \eta(t+1)^2/(8\delta\|A\|_{1\rightarrow 1}^2)$. 
\end{lemma}
\begin{proof}
For $t=1$, we have $\bar{\alpha}^1 = \alpha^1 = 1/(\delta M^1) \geq \eta/(2\delta\|A\|_{1\rightarrow 1}^2)$ where $M^1 \leq 2\eta^{-1}\|A\|_1^2$ has been 
proven in Lemma~\ref{Lemma:Line-search-iteration}. Thus, the desired result holds true when $t=1$. Then we proceed to prove that it holds true for $t \geq 1$ using the induction. Indeed, we have 
\begin{eqnarray*}
\bar{\alpha}^{t+1} & = & \bar{\alpha}^t + \alpha^{t+1} \ = \ \bar{\alpha}^t + \frac{1 + \sqrt{1 + 4\delta M^t\bar{\alpha}^t}}{2\delta M^t} \\
& = & \bar{\alpha}^t + \frac{1}{2\delta M^t} + \sqrt{\frac{1}{4(\delta M^t)^2} + \frac{\bar{\alpha}^t}{\delta M^t}} \\
& \geq & \bar{\alpha}^t + \frac{1}{2\delta M^t} + \sqrt{\frac{\bar{\alpha}^t}{\delta M^t}} \\
& \geq & \bar{\alpha}^t + \frac{\eta}{4\delta\|A\|_{1\rightarrow 1}^2} + \sqrt{\frac{\eta\bar{\alpha}^t}{2\delta\|A\|_{1\rightarrow 1}^2}}, 
\end{eqnarray*}
where the last inequality comes from $M^t \leq 2\eta^{-1}\|A\|_{1\rightarrow 1}^2$ (cf. Lemma~\ref{Lemma:Line-search-iteration}). We assume the desired result holds true for $t = k_0$. Then, we find that 
\begin{eqnarray*}
\bar{\alpha}^{k_0+1} & \geq & \frac{\eta(k_0+1)^2}{8\delta\|A\|_{1\rightarrow 1}^2} + \frac{\eta}{4\delta\|A\|_{1\rightarrow 1}^2} + \sqrt{\frac{\eta^2(k_0+1)^2}{16\delta^2\|A\|_{1\rightarrow 1}^4}} \ = \ \frac{\eta((k_0+1)^2 + 2 + 2(k_0+1))}{8\delta\|A\|_{1\rightarrow 1}^2} \\
& \geq & \frac{\eta(k_0+2)^2}{8\delta\|A\|_{1\rightarrow 1}^2}. 
\end{eqnarray*}
This completes the proof. 
\end{proof}

\subsection{Complexity analysis for the APDAMD algorithm}\label{Sec:complex_APDAMD}
With the key properties of Algorithm~\ref{Algorithm:APDAMD} for the general setup in Eq.~\eqref{prob:generalOT} at hand, we are now ready to analyze the complexity of the APDAMD algorithm for solving the regularized OT problem. Indeed, we set $\varphi(\lambda)$ using Eq.~\eqref{prob:dualregOT-APDAMD} where $\lambda : = (\alpha, \beta)$, given by 
\begin{equation*}
\min_{\alpha, \beta \in \br^n} \ \varphi(\alpha, \beta) \ \mydefn \ \eta\log\left(\sum_{i,j=1}^n e^{-\frac{C_{ij} - \alpha_i - \beta_j}{\eta} -1}\right) - \langle \alpha, r\rangle - \langle \beta, c\rangle. 
\end{equation*}
By means of transformations $u_i = \frac{\alpha_i}{\eta} - \frac{1}{2}$ and $v_j = \frac{\beta_j}{\eta} - \frac{1}{2}$, the objective function is 
\begin{equation*}
\bar{\varphi}(u, v) = \eta\left(\log(\one_n^\top B(u, v) \one_n) - \langle u, r\rangle - \langle v, c\rangle - 1\right). 
\end{equation*}
After the simple calculation, we obtain that an point $(u, v)$ is an optimal solution if it satisfies that $B(u, v)\one_n - r = B(u, v)^\top\one_n - c =\zero_n$; see also~\citet[Lemma~5]{Guminov-2019-Accelerated}. This together with the similar argument from the proof of Lemma~\ref{Lemma:Boundedness} implies that there exists an optimal solution $(u^*, v^*)$ such that $\|u^*\| \leq R$ and $\|v^*\| \leq R$. Therefore, we conclude that there exists an optimal solution $(\alpha^*, \beta^*)$ of the function $\varphi(\alpha, \beta)$ defined by Eq.~\eqref{prob:dualregOT-APDAMD} such that 
\begin{equation}
\label{eq:new_upper_inf_norm}
\|\alpha^*\|_\infty \ \leq \ \eta (R + 1/2), \quad \|\beta^*\|_\infty \ \leq \ \eta (R + 1/2), 
\end{equation}
where $R$ is defined in Lemma~\ref{Lemma:Boundedness}. Then, we proceed to the following key result determining an upper bound for the number of iterations for Algorithm~\ref{Algorithm:APDAMD} to reach a desired accuracy $\varepsilon'$:
\begin{theorem}\label{Theorem:APDAMD-Total-Complexity}
Letting $\{X^t\}_{t \geq 0}$ be the iterates generated by Algorithm~\ref{Algorithm:APDAMD}, the number of iterations required to satisfy $\|A\textnormal{vec}(X^t) - b\|_1 \leq \varepsilon'$ is upper bounded by 
\begin{equation*}
t \ \leq \ 1 + 4\sqrt{2}\|A\|_{1\rightarrow 1}\sqrt{\delta(R + 1/2)/\varepsilon'}, 
\end{equation*}
where $R > 0$ is defined in Lemma~\ref{Lemma:Boundedness}.
\end{theorem}
\begin{proof}
From Lemma~\ref{Lemma:estimate-sequence}, we have
\begin{equation*}
\bar{\alpha}^t\varphi(\lambda^t) \ \leq \ \min_{z \in B_\infty(2\widehat{R})} \left\{\sum_{j=0}^t (\alpha^j (\varphi(\mu^j) + \langle\nabla\varphi(\mu^j), z - \mu^j\rangle)) + \|z\|_\infty^2\right\}, 
\end{equation*}
where $\widehat{R} = \eta(R + 1/2)$ is the upper bound for $\ell_\infty$-norm of optimal solutions of dual regularized OT problem in Eq.~\eqref{prob:dualregOT-APDAMD} and $B_\infty(r) := \{\lambda \in \br^n \mid \|\lambda\|_\infty \leq r\}$. This implies that
\begin{equation*}
\bar{\alpha}^t\varphi(\lambda^t) \leq \min_{z \in B_\infty(2\widehat{R})} \left\{\sum_{j=0}^t (\alpha^j(\varphi(\mu^j) + \langle\nabla\varphi(\mu^j), z - \mu^j\rangle))\right\} + 4 \widehat{R}^2. 
\end{equation*}
Since $\varphi$ is the dual objective function of regularized OT problem, we further have
\begin{eqnarray*}
\varphi(\mu^j) + \langle\nabla\varphi(\mu^j), z - \mu^j\rangle & = & \langle\mu^j, b - Ax(\mu^j)\rangle - f(x(\mu^j)) + \langle z - \mu^j, b - Ax(\mu^j)\rangle \\
& = & - f(x(\mu^j)) + \langle z, b - Ax(\mu^j)\rangle. 
\end{eqnarray*}
Therefore, we conclude that 
\begin{eqnarray*}
\bar{\alpha}^t\varphi(\lambda^t) & \leq & \min_{z \in B_\infty(2\widehat{R})} \left\{\sum_{j=0}^t (\alpha^j(\varphi(\mu^j) + \langle\nabla\varphi(\mu^j), z - \mu^j\rangle))\right\} + 4 \widehat{R}^2 \\
& \leq & 4\widehat{R}^2 -\bar{\alpha}^t f(x^t) + \min_{z \in B_\infty(2\widehat{R})} \left\{\bar{\alpha}^t\langle z, b - Ax^t\rangle\right\} \\
& = & 4\widehat{R}^2 -\bar{\alpha}^t f(x^t) - 2 \bar{\alpha}^t\widehat{R}\| Ax^t - b\|_1, 
\end{eqnarray*}
where the second inequality comes from the convexity of $f$ and the last equality comes from the fact that $\ell_1$-norm is the dual norm of $\ell_\infty$-norm. That is to say, 
\begin{equation*}
f(x^t) + \varphi(\lambda^t) + 2\widehat{R}\|Ax^t - b\|_1 \ \leq \ 4\widehat{R}^2/\bar{\alpha}^t. 
\end{equation*}
Let $\lambda^*$ be an optimal solution to dual regularized OT problem such that $\|\lambda\|_\infty \leq \widehat{R}$, we have
\begin{eqnarray*}
f(x^t) + \varphi(\lambda^t) & \geq & f(x^t) + \varphi(\lambda^*) \ = \ f(x^t) + \langle\lambda^*, b\rangle + \max_{x \in \br^n} \left\{-f(x) -\langle A^\top\lambda^*, x\rangle\right\} \\
& \geq & f(x^t) + \langle\lambda^*, b\rangle - f(x^t) - \langle \lambda^*, Ax^t\rangle \ = \ \langle\lambda^*, b - Ax^t\rangle \\
& \geq & -\widehat{R}\|Ax^t - b\|_1, 
\end{eqnarray*} 
Therefore, we conclude that 
\begin{equation*}
\|Ax^t - b\|_1 \ \leq \ \frac{4\widehat{R}}{\bar{\alpha}^t} \ \leq \ \frac{32\delta(R + 1/2)\|A\|_{1\rightarrow 1}^2}{(t+1)^2}, 
\end{equation*}
which implies the desired result. 
\end{proof}
Now, we are ready to present the complexity bound of Algorithm~\ref{Algorithm:ApproxOT_APDAMD} for approximating the OT problem.
\begin{theorem} \label{Theorem:ApproxOT-APDAMD-Total-Complexity}
The APDAMD algorithm for approximating optimal transport (Algorithm~\ref{Algorithm:ApproxOT_APDAMD}) returns $\hat{X} \in \br^{n \times n}$ satisfying $\hat{X}\one_n = r$, $\hat{X}^\top\one_n = c$ and Eq.~\eqref{Criteria:Approximation} in
\begin{equation*}
\bigO\left(\frac{n^2 \sqrt{\delta} \|C\|_\infty\log(n)}{\varepsilon}\right)
\end{equation*}
arithmetic operations. 
\end{theorem}
\begin{proof}
Using the same argument as in Theorem~\ref{Theorem:ApproxOT-Greenkhorn-Total-Complexity}, we have
\begin{equation*}
\langle C, \hat{X}\rangle - \langle C, X^*\rangle \ \leq \ \varepsilon/2 + 4(\|\tilde{X}\one_n - r\|_1 + \|\tilde{X}^\top\one_n - c\|_1)\|C\|_\infty,
\end{equation*}
where $\hat{X}$ is returned by Algorithm~\ref{Algorithm:ApproxOT_APDAMD}, $X^*$ is a solution to the optimal transport problem and $\tilde{X}$ is returned by Algorithm~\ref{Algorithm:APDAMD} with $\tilde{r}$, $\tilde{c}$ and $\varepsilon'/2$ in Step 3 of Algorithm~\ref{Algorithm:ApproxOT_APDAMD}. Note that $\|\tilde{X}\one_n - r\|_1 + \|\tilde{X}^\top\one_n - c\|_1 \leq \varepsilon'$. Thus, we have $\langle C, \hat{X}\rangle - \langle C, X^*\rangle \leq \varepsilon$. The remaining step is to analyze the complexity bound. Since $\|A\|_{1\rightarrow 1}$ equals to the maximum $\ell_1$-norm of a column of $A$ and each column of $A$ contains only two nonzero elements which are equal to one, we have $\|A\|_{1\rightarrow 1}=2$. This together with Lemma~\ref{Lemma:Line-search-iteration} and Theorem~\ref{Theorem:APDAMD-Total-Complexity} yields that 
\begin{eqnarray*}
N_t & \leq & 4t + 4 + \frac{2\log(\|A\|_{1\rightarrow 1}^2/(2\eta)) - 2\log(L^0)}{\log 2} \\
& \leq & 8 + 16\sqrt{2}\|A\|_{1\rightarrow 1}\sqrt{\frac{\delta(R + 1/2)}{\varepsilon'}} + \frac{2\log(\|A\|_{1\rightarrow 1}^2/(2\eta)) - 2\log(L^0)}{\log 2} \\
& = & 8 + 256\sqrt{\frac{\eta(R + 1/2)\|C\|_\infty\log(n)}{\varepsilon}} + \frac{2\log(\log(n)/\varepsilon)}{\log 2}.
\end{eqnarray*}
Using Lemma~\ref{Lemma:Boundedness} and the definition of $\tilde{r}$ and $\tilde{c}$ in Algorithm~\ref{Algorithm:ApproxOT_APDAMD}, we have 
\begin{equation*}
R \ \leq \ \frac{4\|C\|_\infty\log(n)}{\varepsilon} + \log(n) - 2\log\left(\frac{\varepsilon}{64n\|C\|_\infty}\right).
\end{equation*}
Therefore, we conclude that 
\begin{eqnarray*}
N_t & \leq & 256\sqrt{\frac{\delta\|C\|_\infty\log(n)}{\varepsilon}}\sqrt{\frac{4\|C\|_\infty\log(n)}{\varepsilon} + \log(n) - 2\log\left(\frac{\varepsilon}{64n\|C\|_\infty}\right) + \frac{1}{2}} \\ 
& & + \frac{2\log(\log(n)/\varepsilon)}{\log 2} + 8 \ = \ \bigO\left(\frac{\sqrt{\delta}\left\|C\right\|_\infty\log(n)}{\varepsilon}\right). 
\end{eqnarray*}
The total iteration complexity in Step 3 of Algorithm~\ref{Algorithm:ApproxOT_APDAMD} is $\bigO(\sqrt{\delta}\|C\|_\infty\log(n)\varepsilon^{-1})$. Each iteration of Algorithm~\ref{Algorithm:APDAMD} requires $O(n^2)$ arithmetic operations. Thus, the total number of arithmetic operations is $\bigO(n^2\sqrt{\delta}\|C\|_\infty\log(n)\varepsilon^{-1})$. Furthermore, $\tilde{r}$ and $\tilde{c}$ in Step 2 of Algorithm~\ref{Algorithm:ApproxOT_APDAMD} can be found in $\bigO(n)$ arithmetic operations and \citet[Algorithm~2]{Altschuler-2017-Near} requires $\bigO(n^2)$ arithmetic operations. Therefore, we conclude that the total number of arithmetic operations is $\bigO(n^2\sqrt{\delta}\|C\|_\infty\log(n)\varepsilon^{-1})$.
\end{proof}
The complexity bound of the APDAMD algorithm in Theorem~\ref{Theorem:ApproxOT-APDAMD-Total-Complexity} suggests an interesting feature of the (regularized) OT problem. Indeed, the dependence of that bound on $\delta$ manifests the necessity of $\ell_\infty$-norm in the understanding of the complexity of the regularized OT problem. This view is also in harmony with the proof technique of running time for the Greenkhorn algorithm in Section~\ref{sec:greenkhorn}, where we rely on $\ell_\infty$-norm of optimal solutions of the dual regularized OT problem to measure the progress in the objective value among the successive iterates. 
\begin{algorithm}[!t]
\caption{Approximating OT by~\citet[Algorithm~3]{Dvurechensky-2018-Computational}}\label{Algorithm:ApproxOT_APDAGD}
\begin{algorithmic}
\STATE \textbf{Input:} $\eta = \varepsilon/(4\log(n))$ and $\varepsilon'=\varepsilon/(8\|C\|_\infty)$. 
\STATE \textbf{Step 1:} Let $\tilde{r} \in \Delta_n$ and $\tilde{c} \in \Delta_n$ be defined by $(\tilde{r}, \tilde{c}) = (1 - \varepsilon'/8)(r, c) + (\varepsilon'/8n)(\one_n, \one_n)$. 
\STATE \textbf{Step 2:} Let $A \in \br^{2n \times n^2}$ and $b \in \br^{2n}$ be defined by $A\text{vec}(X) = \begin{pmatrix} X\one_n \\ X^\top\one_n \end{pmatrix}$ and $b = \begin{pmatrix} \tilde{r} \\ \tilde{c} \end{pmatrix}$. 
\STATE \textbf{Step 3:} Compute $\tilde{X} = \text{APDAGD}(\varphi, A, b, \varepsilon'/2)$ where $\varphi$ is defined by Eq.~\eqref{prob:dualregOT-old}. 
\STATE \textbf{Step 4:} Round $\tilde{X}$ to $\hat{X}$ using~\citet[Algorithm~2]{Altschuler-2017-Near} such that $\hat{X}\one_n = r$ and $\hat{X}^\top\one_n = c$. 
\end{algorithmic}
\end{algorithm}
\subsection{Revisiting the APDAGD algorithm}
In this section, we revisit the APDAGD algorithm for the regularized OT problem~\citep{Dvurechensky-2018-Computational}. First, we point out that the complexity bound of the APDAGD algorithm for regularized OT is not $\bigOtil(\min\{n^{9/4}\varepsilon^{-1}, n^{2}\varepsilon^{-2}\})$ as claimed from their theoretical analysis. This is confirmed by a simple counterexample. We further provide a new complexity bound of the APDAGD algorithm using our techniques in Section~\ref{Sec:complex_APDAMD}. Despite the issue with regularized OT, we wish to emphasize that the APDAGD algorithm is still an interesting and efficient accelerated algorithm for general problem in Eq.~\eqref{prob:generalOT} with theoretical guarantee under the certain conditions. More precisely, while \citet[Theorem~3]{Dvurechensky-2018-Computational} is not applicable to regularized OT since there exists no dual solution with a constant bound in $\ell_2$-norm, this theorem is valid and can be used for other regularized problems with bounded optimal dual solution. 

To facilitate the ensuing discussion, we first present the complexity bound for regularized OT in \citet{Dvurechensky-2018-Computational} using the notation from the current paper. Indeed, we recall that the APDAGD algorithm is developed for solving the optimization problem with the function $\varphi$ defined by Eq.~\eqref{prob:dualregOT-old} as follows,   
\begin{equation}\label{prob:dualOT-new}
\min_{\alpha, \beta \in \br^n} \ \varphi(\alpha, \beta) \ = \ \eta\left(\sum_{i,j=1}^n e^{-\frac{C_{ij} - \alpha_i - \beta_j}{\eta} -1}\right) - \langle \alpha, r\rangle - \langle \beta, c\rangle. 
\end{equation}
\begin{theorem}[Theorem 4 in \citet{Dvurechensky-2018-Computational}] \label{theorem:org_complex_APDAGD}
The APDAGD algorithm for approximating optimal transport returns
$\hat{X} \in \br^{n \times n}$ satisfying $\hat{X}\one_n = r$, $\hat{X}^\top\one_n = c$ and~\eqref{Criteria:Approximation} in a number of arithmetic operations bounded as
\begin{equation*}
\bigO\left(\min \left\{\frac{n^{9/4} \sqrt{\overline{R} \left\|C\right\|_
\infty\log(n)}}{\varepsilon}, \frac{n^{2} \overline{R} \left\|C
\right\|_\infty\log(n)}{\varepsilon^2}\right\} \right),
\end{equation*}
where $\|(\alpha^*, \beta^*)\|_2 \leq \overline{R}$ and $(\alpha^*, \beta^*)$ denotes an optimal solution pair for the function $\widetilde{\varphi}$. 
\end{theorem}
This theorem suggests that the complexity bound is at the order $\bigOtil(\min\{n^{9/4}\varepsilon^{-1}, n^2\varepsilon^{-2}\})$. However, there are two issues: (i) the upper bound $\overline{R}$ is assumed to be bounded and independent of $n$, which is incorrect; see our counterexample in Proposition~\ref{proposition:tight_upper}; (ii) the upper bound $\overline{R}$ is based on $\min_{1 \leq i, j \leq n} \left\{r_i, l_j\right\}$ (cf. Lemma~\ref{Lemma:Boundedness} or~\citet[Lemma 1]{Dvurechensky-2018-Computational}). This implies that the valid algorithm needs to take the rounding error with $r$ and $l$ into account.

\paragraph{Corrected upper bound $\overline{R}$.} Using the similar arguments for deriving the upper bounds from~\eqref{eq:new_upper_inf_norm}, we obtain that an upper bound for $\overline{R}$ is $\bigOtil(n^{1/2})$. The following proposition shows that $\overline{R}$ is indeed $\Omega(n^{1/2})$ for any $\varepsilon \in (0, 1)$. 
\begin{proposition} \label{proposition:tight_upper}
Assume that $C = \one_n\one_n^\top$ and $r = c = (1/n)\one_n$. Given $\varepsilon \in \left(0, 1\right)$ and the regularization term $\eta = \varepsilon/(4\log(n))$, all the optimal solutions of the dual regularized OT problem in Eq.~\eqref{prob:dualOT-new} satisfy that $\|(\alpha^*, \beta^*)\| \gtrsim \sqrt{n}$.
\end{proposition}
\begin{proof}
By the definition $r$, $c$ and $\eta$, we rewrite the dual function $\varphi(\alpha, \beta)$ as follows:
\begin{equation*}
\varphi(\alpha, \beta) \ = \ \frac{\varepsilon}{4e\log(n)}\sum_{1 \leq i, j \leq n} e^{- \frac{4\log(n)(1 - \alpha_i - \beta_j)}{\varepsilon}} - \frac{\one_n^\top\alpha}{n} - \frac{\one_n^\top\beta}{n}.
\end{equation*}
Since $(\alpha^*, \beta^*)$ is the optimal solution of dual regularized OT problem, we have
\begin{equation}\label{example-opt-dualOT}
e^{\frac{4\log(n)\alpha_i^*}{\varepsilon}} \sum_{j=1}^n e^{-\frac{4\log(n)(1 - \beta_j^*)}{\varepsilon}} \ = \ e^{\frac{4\log(n)\beta_i^*}{\varepsilon}} \sum_{j=1}^n e^{- \frac{4\log(n)(1 - \alpha_j^*)}{\varepsilon}} \ = \ \frac{e}{n} \quad \text{for all } i \in [n]. 
\end{equation}
This implies $\alpha_i^* = \alpha_j^*$ and $\beta_i^* = \beta_j^*$ for all $i, j \in [n]$. Thus, we let $A = e^{4\log(n)\alpha_i^*/\varepsilon}$ and $B = e^{4\log(n)\beta_i^*/\varepsilon}$ for the simplicity. Using Eq.~\eqref{example-opt-dualOT}, we have $AB e^{-4\log(n)/\varepsilon} = en^{-2}$ which implies that $AB = e^{\frac{4\log(n)}{\varepsilon} + 1}n^{-2}$. So we have
\begin{equation*}
\alpha_i^* + \beta_i^* = \frac{\varepsilon\left(\log(A) + \log(B)\right)}{4\log(n)} = \frac{\varepsilon}{4\log(n)}\left(\frac{4\log(n)}{\varepsilon} + 1 - 2\log(n)\right) = 1 + \frac{\varepsilon}{4\log(n)} - \frac{\varepsilon}{2}. 
\end{equation*}
Therefore, we conclude that 
\begin{align*}
\|(\alpha^*, \beta^*)\| \geq \sqrt{\frac{\sum_{i = 1}^{n} (\alpha_i^* + \beta_i^*)^2}{2}} = \sqrt{\frac{n}{2}}\left(1 + \frac{\varepsilon}{4\log(n)} - \frac{\varepsilon}{2}\right) \gtrsim \sqrt{n}. 
\end{align*}
As a consequence, we achieve the conclusion of the proposition.
\end{proof}
\paragraph{Approximation algorithm for OT by APDAGD.} We notice that~\citet[Algorithm~4]{Dvurechensky-2018-Computational} lacks the rounding procedure and needs to improved to Algorithm~\ref{Algorithm:ApproxOT_APDAGD}. Here,~\citet[Algorithm~3]{Dvurechensky-2018-Computational} is used in \textbf{Step 3} of Algorithm~\ref{Algorithm:ApproxOT_APDAGD}. Given the corrected upper bound $\overline{R}$ and Algorithm~\ref{Algorithm:ApproxOT_APDAGD} for approximating OT, 
we provide a new complexity bound of Algorithm~\ref{Algorithm:ApproxOT_APDAGD} in the following proposition. 
\begin{proposition} \label{prop:correct_complex_APDAGD}
The APDAGD algorithm for approximating optimal transport (Algorithm~\ref{Algorithm:ApproxOT_APDAGD}) returns $\hat{X} \in \br^{n \times n}$ satisfying $\hat{X}\one_n = r$, $\hat{X}^\top\one_n = c$ and Eq.~\eqref{Criteria:Approximation} in
\begin{equation*}
\bigO\left(\frac{n^{5/2}\|C\|_\infty\sqrt{\log(n)}}{\varepsilon}\right)
\end{equation*}
arithmetic operations.
\end{proposition}
\begin{proof}
The proof of Proposition~\ref{prop:correct_complex_APDAGD} is a modification of the proof for~\citet[Theorem 4]{Dvurechensky-2018-Computational}. Therefore, we only give a proof sketch to ease the presentation. More specifically, we follow the argument of~\citet[Theorem 4]{Dvurechensky-2018-Computational} and obtain that the number of iterations for Algorithm~\ref{Algorithm:ApproxOT_APDAGD} required to reach the tolerance $\varepsilon$ is
\begin{equation}\label{inequality-proposition-APDAGD}
t \ \leq \ \max \left\{\bigO\left(\min \left\{\frac{n^{1/4} 
\sqrt{\overline{R} \left\|C\right\|_\infty\log(n)}}{\varepsilon}, 
\frac{\overline{R} \left\|C\right\|_\infty\log(n)}{\varepsilon^2}
\right\} \right), \bigO \left(\frac{\overline{R} \sqrt{\log n}}
{\varepsilon}\right)\right\}.
\end{equation}
Furthermore, we have $\overline{R} \leq \sqrt{n}\eta R$ where $R = \eta^{-1}\|C\|_\infty + \log(n) - 2\log(\min_{1 \leq i, j \leq n} \{r_i, c_j\})$ (cf. Lemma~\ref{Lemma:Boundedness}). Therefore, the total iteration complexity in Step 3 of Algorithm~\ref{Algorithm:ApproxOT_APDAGD} is $\bigO(\sqrt{n\log(n)}\|C\|_\infty\varepsilon^{-1})$. Each iteration of the APDAGD algorithm requires $\bigO(n^2)$ arithmetic operations. Thus, the total number of arithmetic operations is $\bigO(n^{5/2}\|C\|_\infty\sqrt{\log(n)}\varepsilon^{-1})$. Furthermore, $\tilde{r}$ and $\tilde{c}$ in \textbf{Step 2} of Algorithm~\ref{Algorithm:ApproxOT_APDAGD} can be found in $\bigO(n)$ arithmetic operations and~\citet[Algorithm~2]{Altschuler-2017-Near} requires $\bigO(n^2)$ arithmetic operations. Therefore, we conclude that the total number of arithmetic operations is $\bigO(n^{5/2}\|C\|_\infty\sqrt{\log(n)}\varepsilon^{-1})$.
\end{proof}
\begin{remark}
As indicated in Proposition~\ref{prop:correct_complex_APDAGD}, the corrected 
complexity bound of APDAGD algorithm for the regularized OT is similar to that of our APDAMD algorithm when we choose $\phi = (1/2n)\|\cdot\|^2$ and have $\delta = n$. From this perspective, our APDAMD algorithm can be viewed as a generalization of the APDAGD algorithm. Since our APDAMD algorithm utilizes $\ell_\infty$-norm in the line search criterion, it is more robust than the APDAGD algorithm in practice; see the next section for the details. 
\end{remark}

\section{Experiments}\label{sec:experiments}
In this section, we conduct the extensive comparative experiments with the Greenkhorn and APDAMD algorithms on both synthetic images and real images from MNIST Digits dataset\footnote{http://yann.lecun.com/exdb/mnist/}. The baseline algorithms include the Sinkhorn algorithm~\citep{Cuturi-2013-Sinkhorn, Altschuler-2017-Near}, the APDAGD algorithm~\citep{Dvurechensky-2018-Computational} and the GCPB algorithm~\citep{Genevay-2016-Stochastic}. The Greenkhorn and APDAGD algorithms have been shown outperform the Sinkhorn algorithm in~\citet{Altschuler-2017-Near} and~\citet{Dvurechensky-2018-Computational}, respectively. However, we repeat some of these comparisons to ensure that our evaluation is systematic and complete. Finally, we utilize the default linear programming solver in MATLAB to obtain the optimal value of the unregularized optimal transport problem.
\begin{figure*}[!t]
\begin{minipage}[b]{.5\textwidth}
\includegraphics[width=75mm,height=55mm]{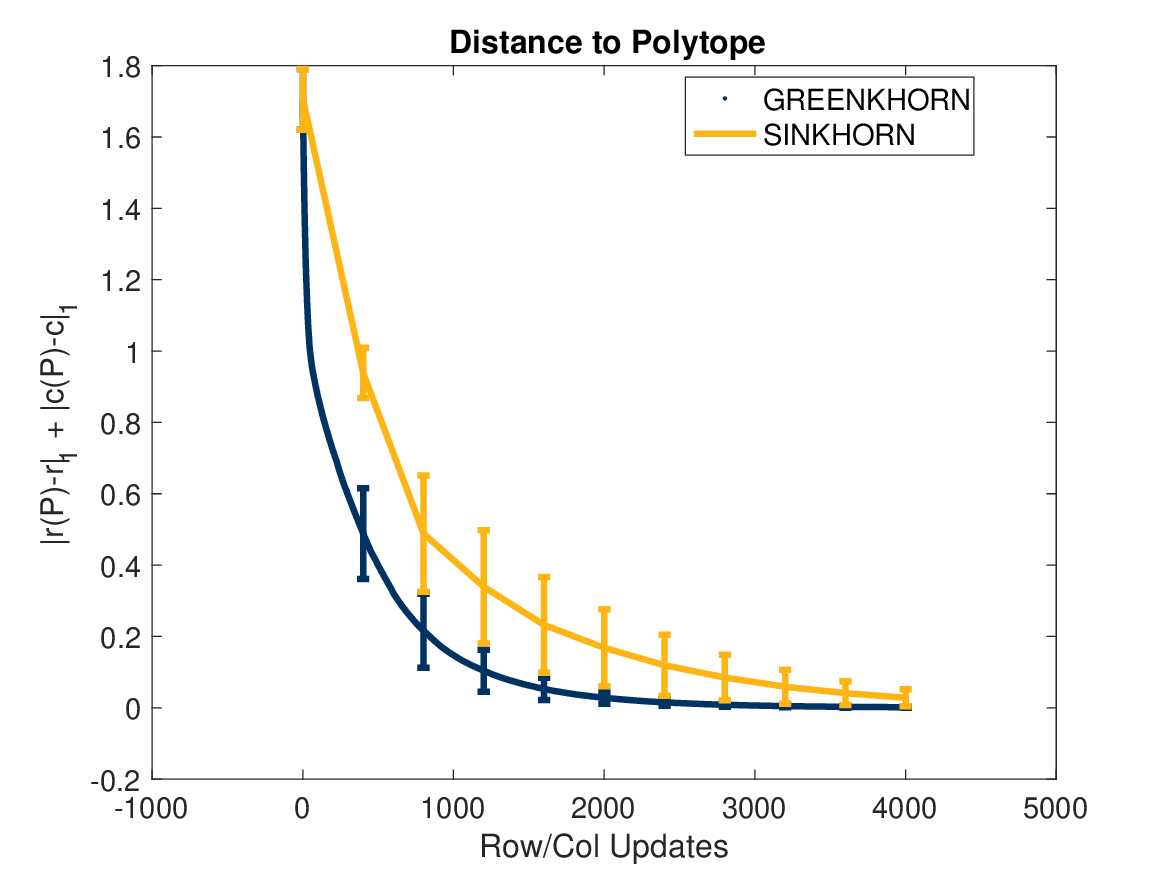}
\end{minipage}
\quad
\begin{minipage}[b]{.5\textwidth}
\includegraphics[width=75mm,height=55mm]{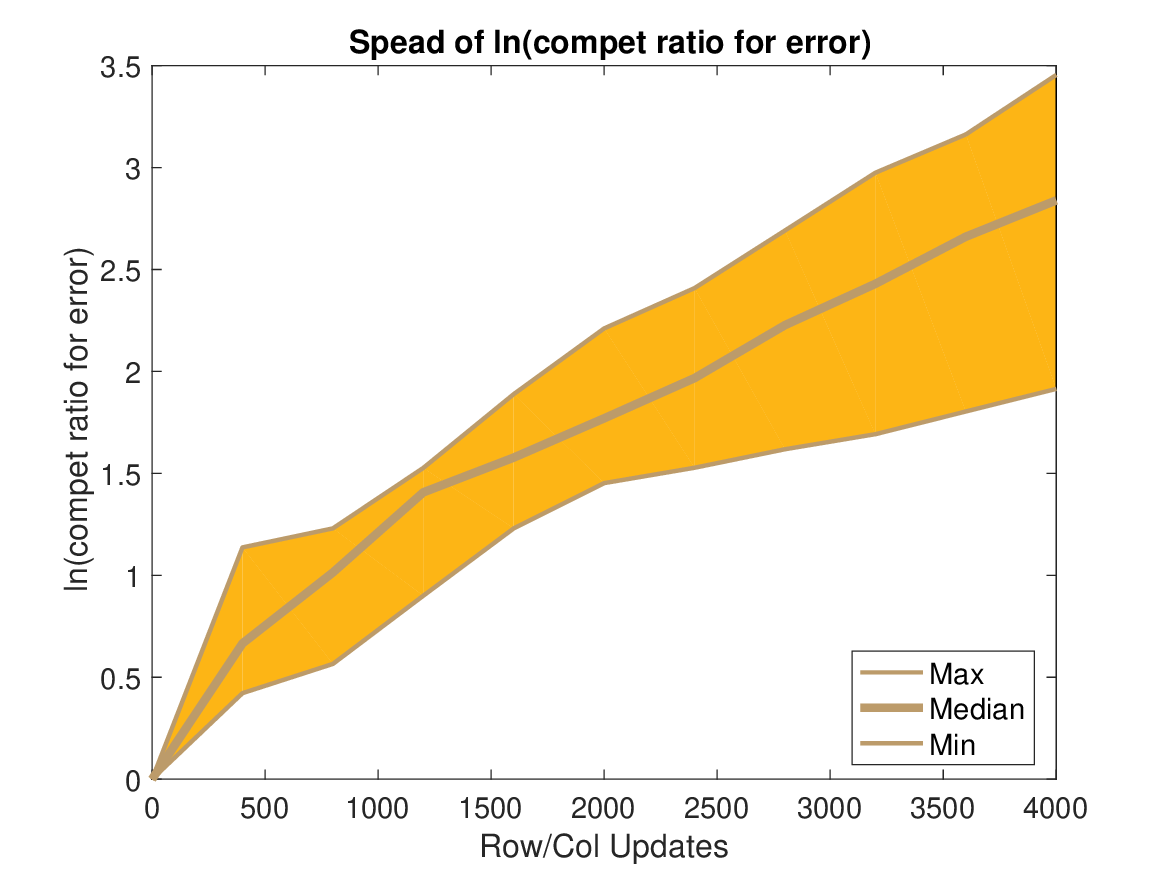}
\end{minipage}
\\
\begin{minipage}[b]{.5\textwidth}
\includegraphics[width=75mm,height=55mm]{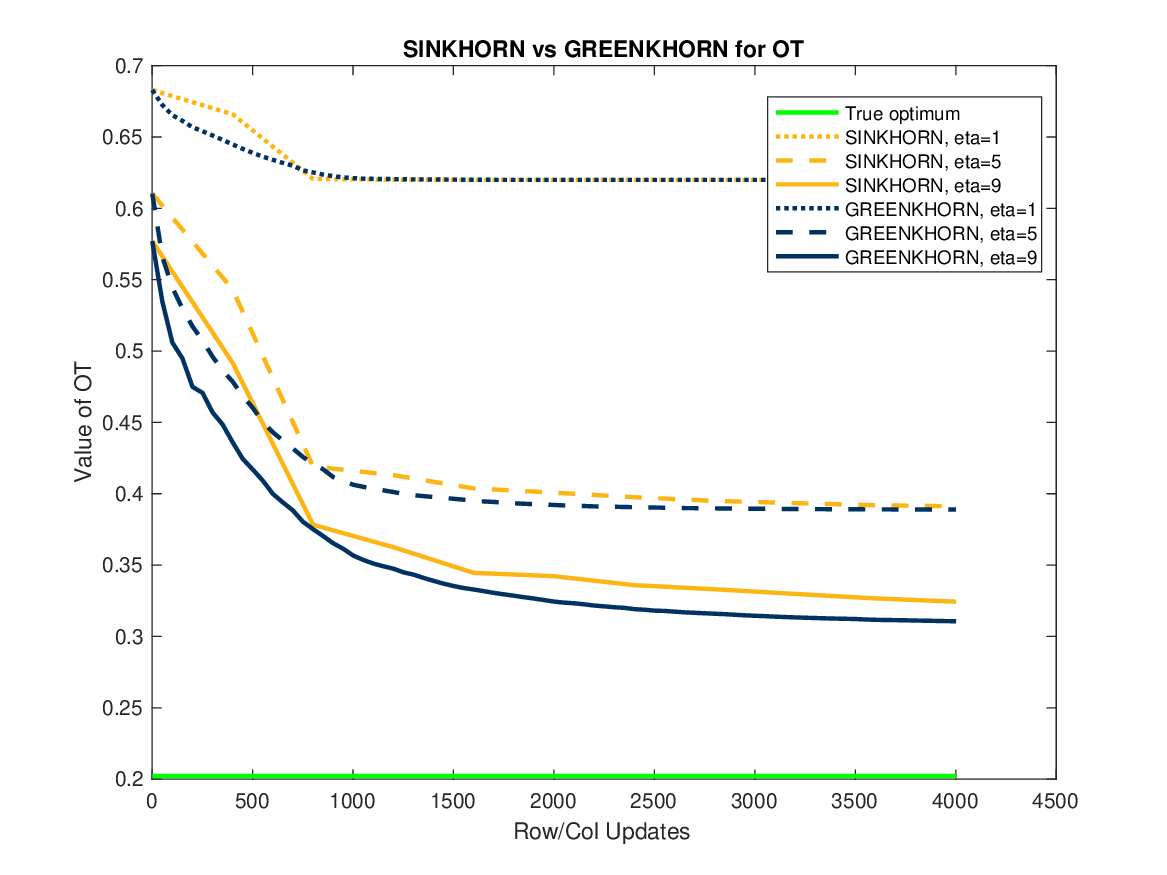}
\end{minipage}
\quad
\begin{minipage}[b]{.5\textwidth}
\includegraphics[width=75mm,height=55mm]{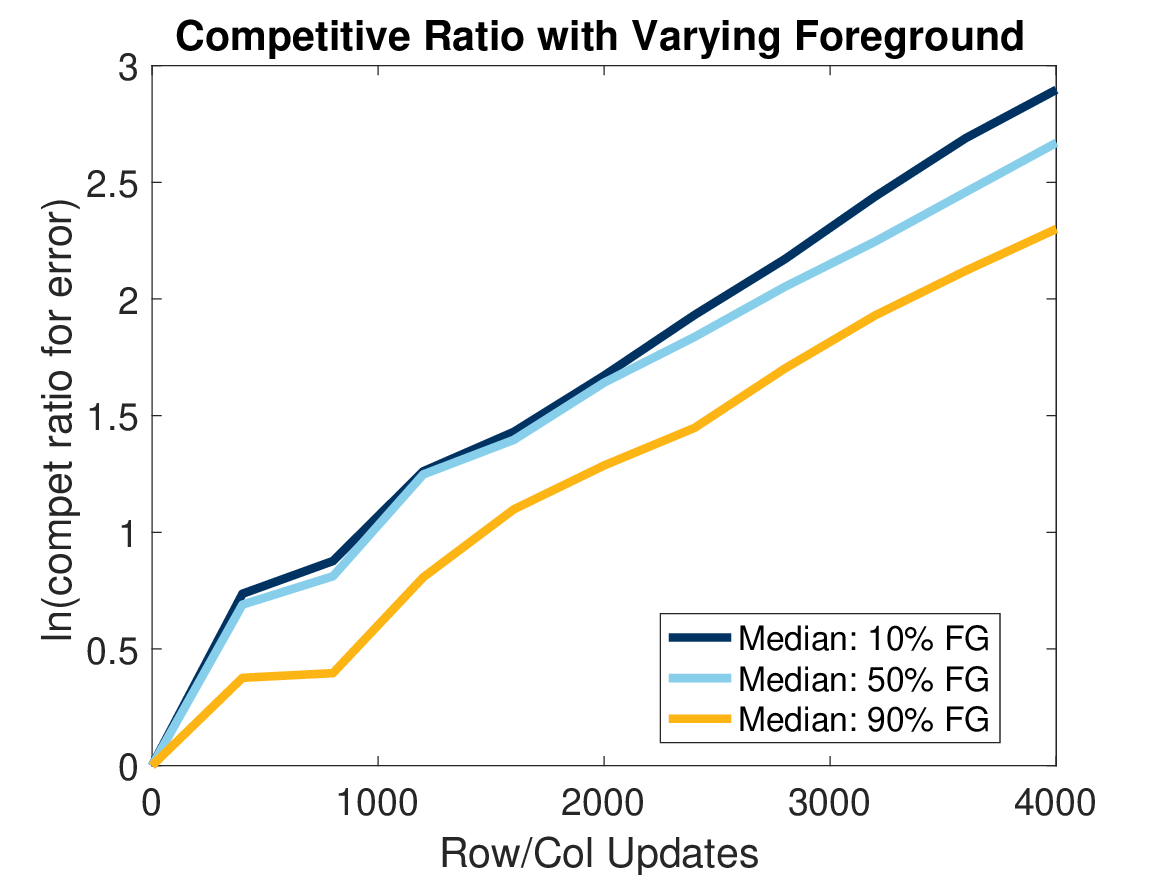}
\end{minipage}
\caption{Performance of the Sinkhorn and Greenkhorn algorithms on the synthetic images. In the top two images, the comparison is based on using the distance $d(P)$ to the transportation polytope, and the maximum, median and minimum of competitive ratios $\log(d(P_{S})/d(P_{G}))$ on ten random pairs of images. Here, $d(P_{S})$ and $d(P_{G})$ refer to the Sinkhorn and Greenkhorn algorithms, respectively. In the bottom left image, the comparison is based on varying the regularization parameter $\eta \in \{1, 5, 9\}$ and reporting the optimal value of the unregularized optimal transport problem. Note that the foreground covers 10\% of the synthetic images here. In the bottom right image, we compare by using the median of competitive ratios with varying coverage ratio of foreground in the range of 10\%, 50\%, and 90\% of the images.}
\label{fig:green_sink_synthetic}
\end{figure*}
\subsection{Synthetic images}
We follow the setup in~\citet{Altschuler-2017-Near} in order to compare different algorithms on the synthetic images. In particular, the transportation distance is defined between a pair of randomly generated synthetic images and the cost matrix is comprised of $\ell_1$ distances among pixel locations in the images.  
\paragraph{Image generation:} Each of the images is of size 20 by 20 pixels and is generated based on randomly positioning a foreground square in otherwise black background. We utilize a uniform distribution on $[0, 1]$ for the intensities of the background pixels and a uniform distribution on $[0, 50]$ for the foreground pixels. To further evaluate the robustness of all the algorithms to the ratio with varying foreground, we vary the proportion of the size of this square in $\{0.1, 0.5, 0.9\}$ of the images and implement all the algorithms on different kind of synthetic images. 
\begin{figure*}[!t]
\begin{minipage}[b]{.5\textwidth}
\includegraphics[width=75mm,height=55mm]{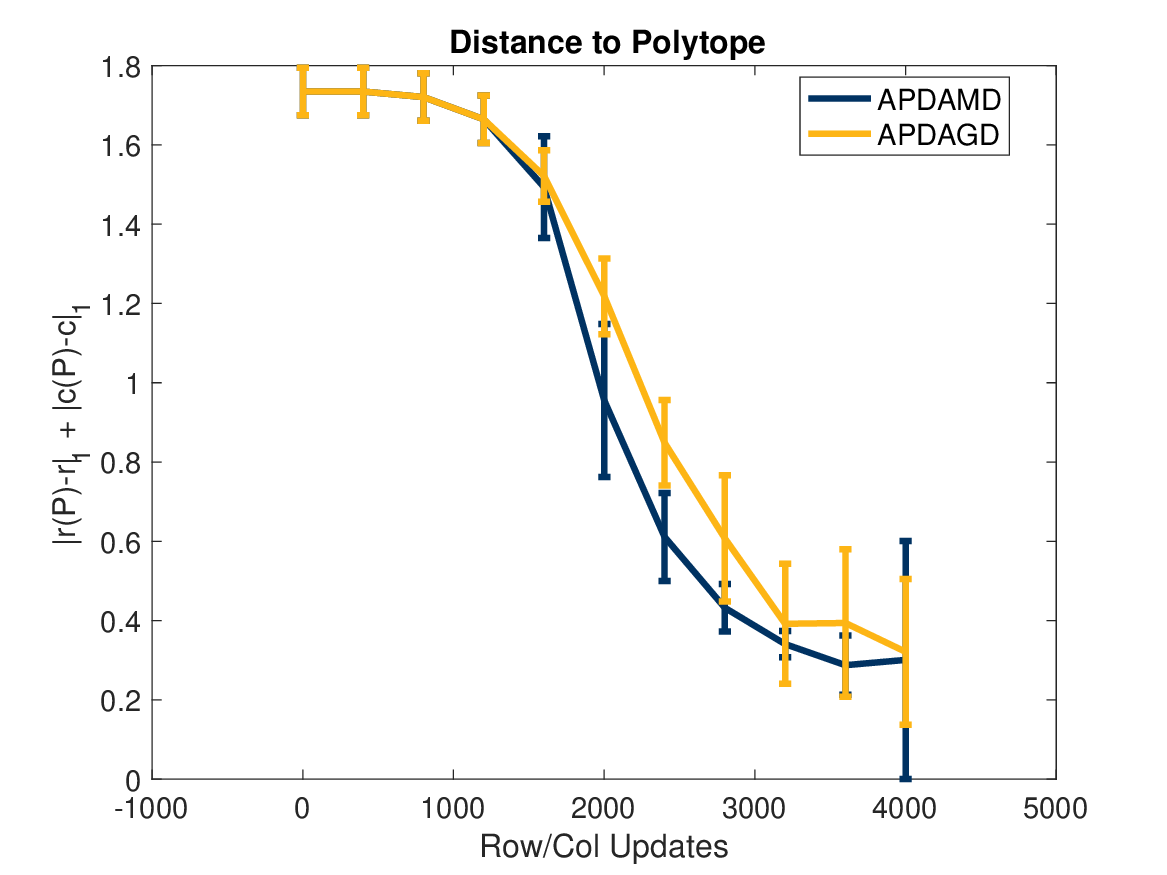}
\end{minipage}
\quad
\begin{minipage}[b]{.5\textwidth}
\includegraphics[width=75mm,height=55mm]{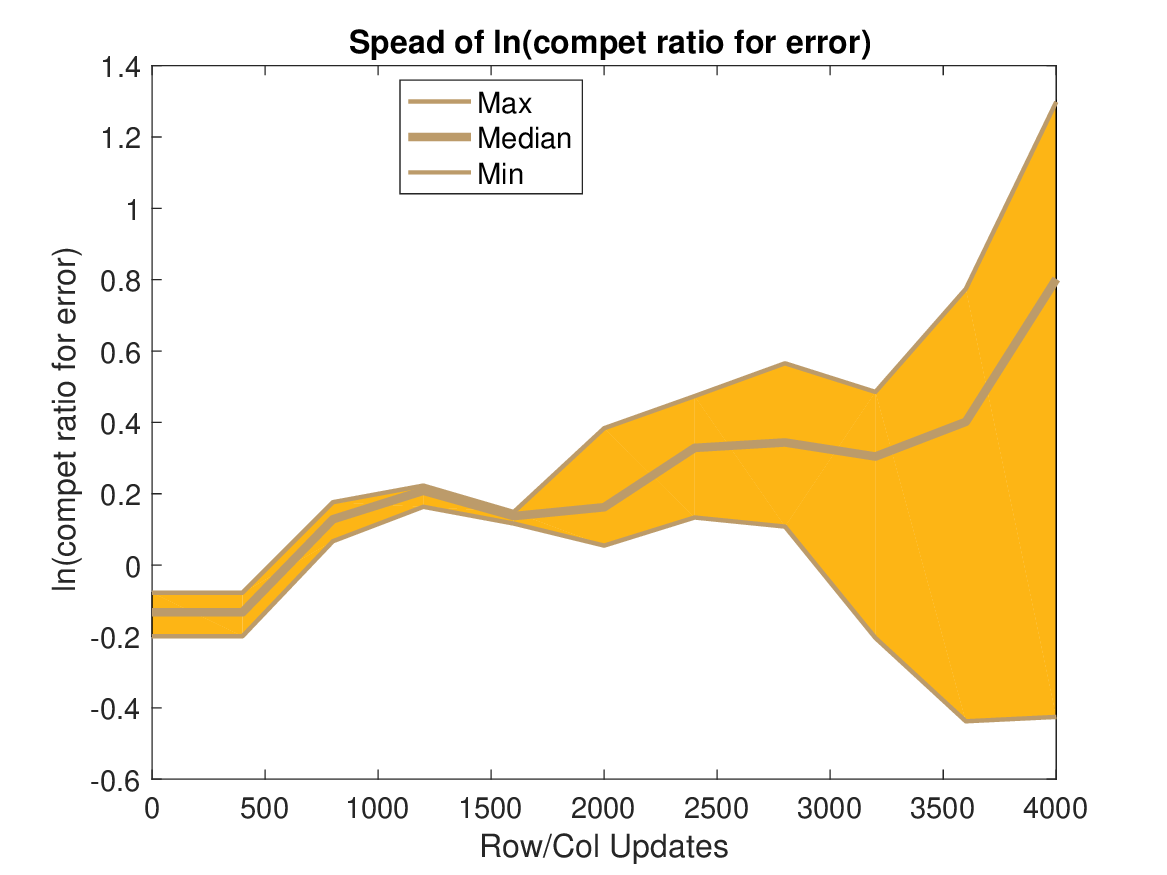}
\end{minipage}
\\ 
\begin{minipage}[b]{.5\textwidth}
\includegraphics[width=75mm,height=55mm]{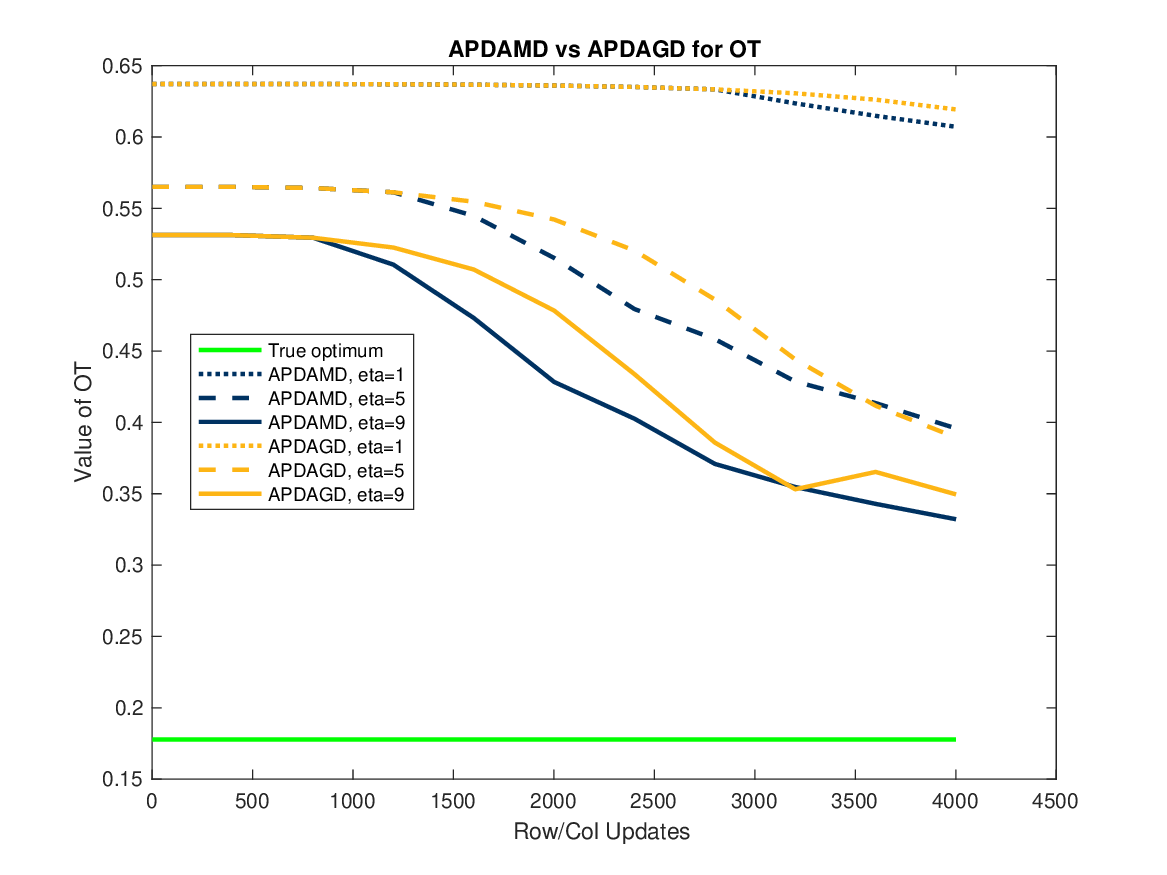}
\end{minipage}
\quad
\begin{minipage}[b]{.5\textwidth}
\includegraphics[width=75mm,height=55mm]{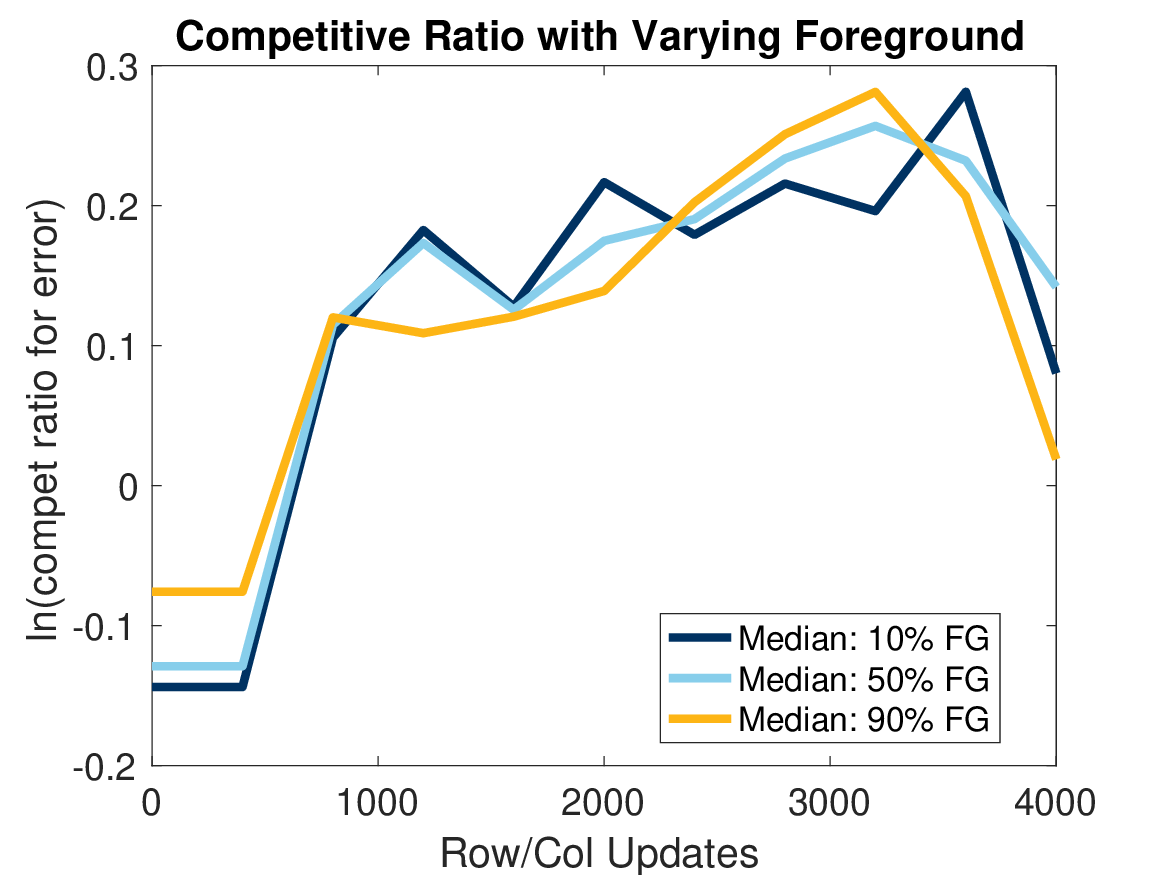}
\end{minipage}
\caption{Performance of the APDAGD and APDAMD algorithms in the synthetic images. All the four images correspond to that in Figure~\ref{fig:green_sink_synthetic}, showing that the APDAMD algorithm is faster and more robust than the APDAGD algorithm. Note that $\log(d(P_{GD})/d(P_{MD}))$ on 10 random pairs of images is consistently used, where $d(P_{GD})$ and $d(P_{MD})$ are for the APDAGD and APDAMD algorithms, respectively.
}\label{fig:amd_agd_synthetic}
\end{figure*}
\paragraph{Evaluation metrics:} Two metrics proposed by~\citet{Altschuler-2017-Near} are used here to quantitatively measure the performance of different algorithms. The first metric is the distance between the output of the algorithm, $X$, and the transportation polytope, i.e., $d(X) = \|r(X) - r\|_1 + \|c(X) - c\|_1$, where $r(X)$ and $c(X)$ are the row and column marginal vectors of the output $X$ while $r$ and $l$ stand for the true row and column marginal vectors. The second metric is the competitive ratio, defined by $\log(d(X_1)/d(X_2))$ where $d(X_1)$ and $d(X_2)$ refer to the distance between the outputs of two algorithms and the transportation polytope.
\begin{figure*}[!t]
\begin{minipage}[b]{.5\textwidth}
\includegraphics[width=75mm,height=55mm]{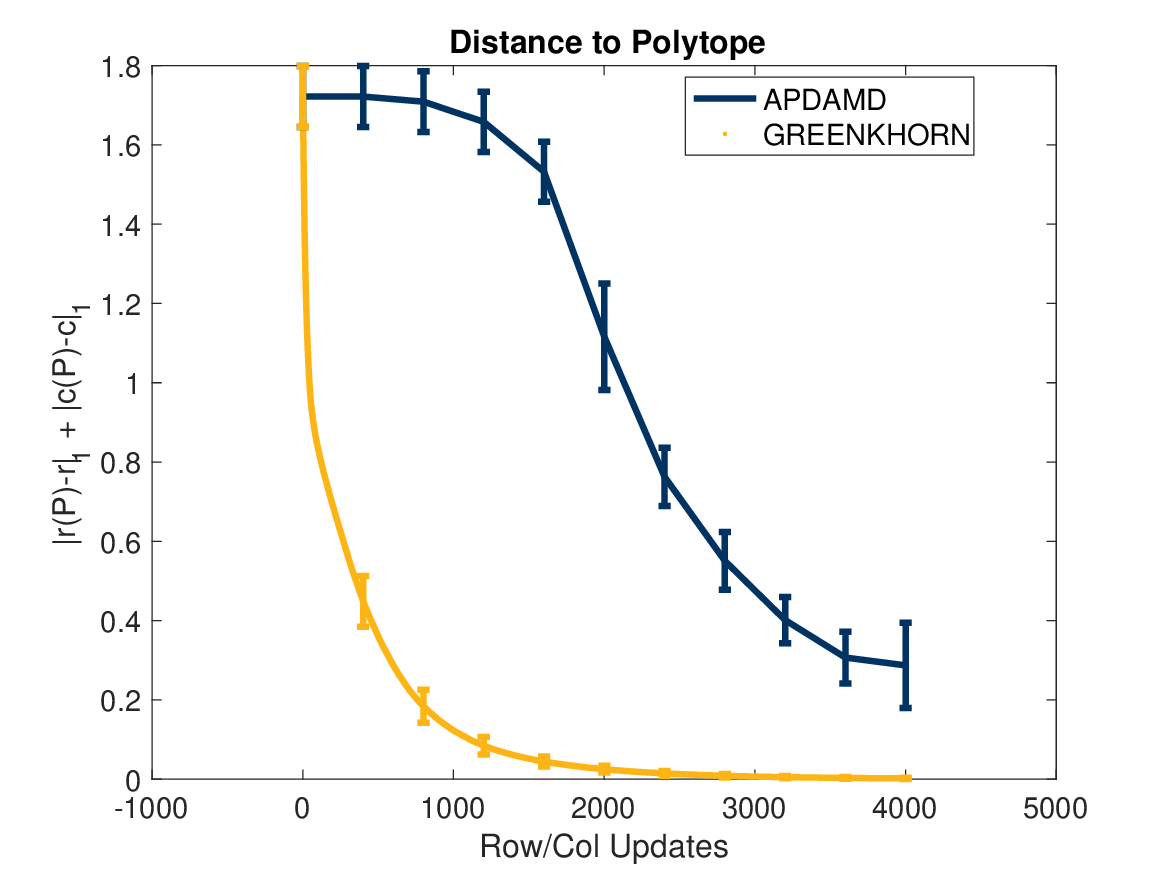}
\end{minipage}
\quad
\begin{minipage}[b]{.5\textwidth}
\includegraphics[width=75mm,height=55mm]{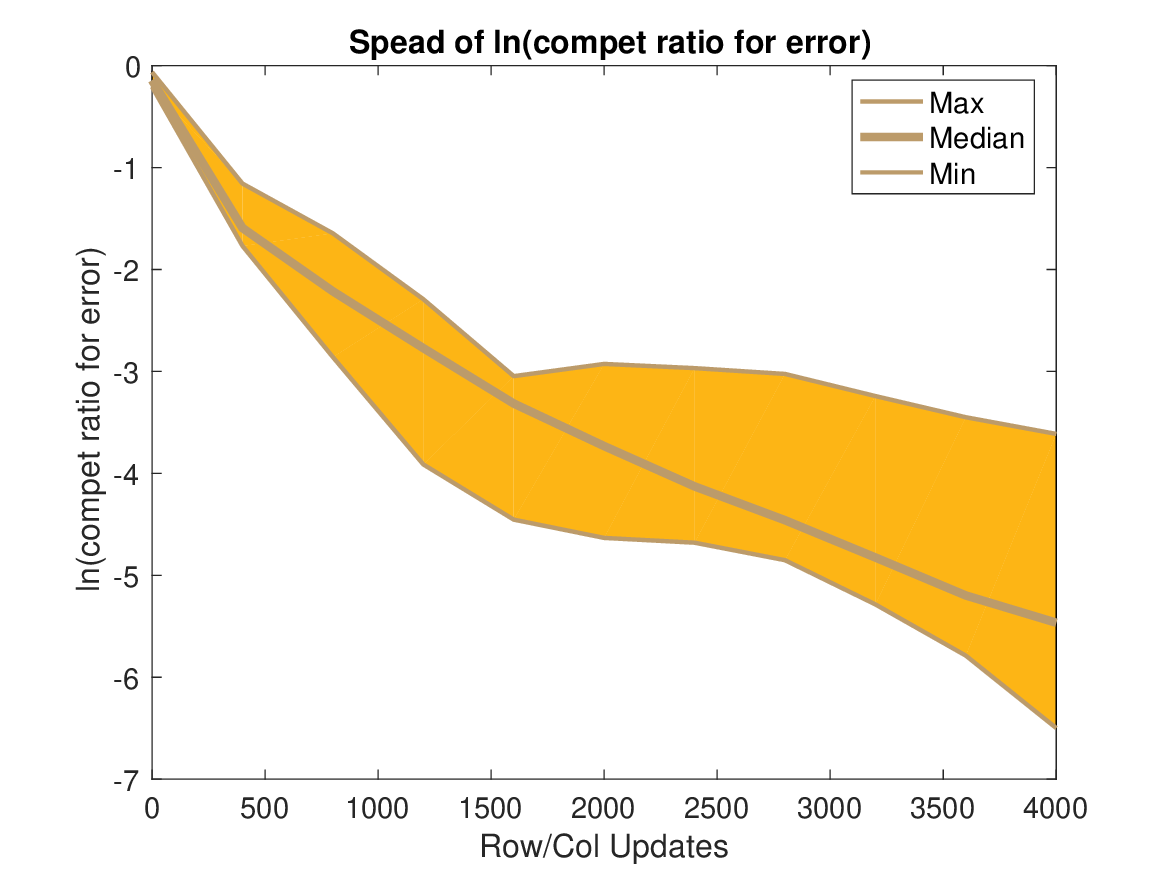}
\end{minipage}
\\ 
\begin{minipage}[b]{.5\textwidth}
\includegraphics[width=75mm,height=55mm]{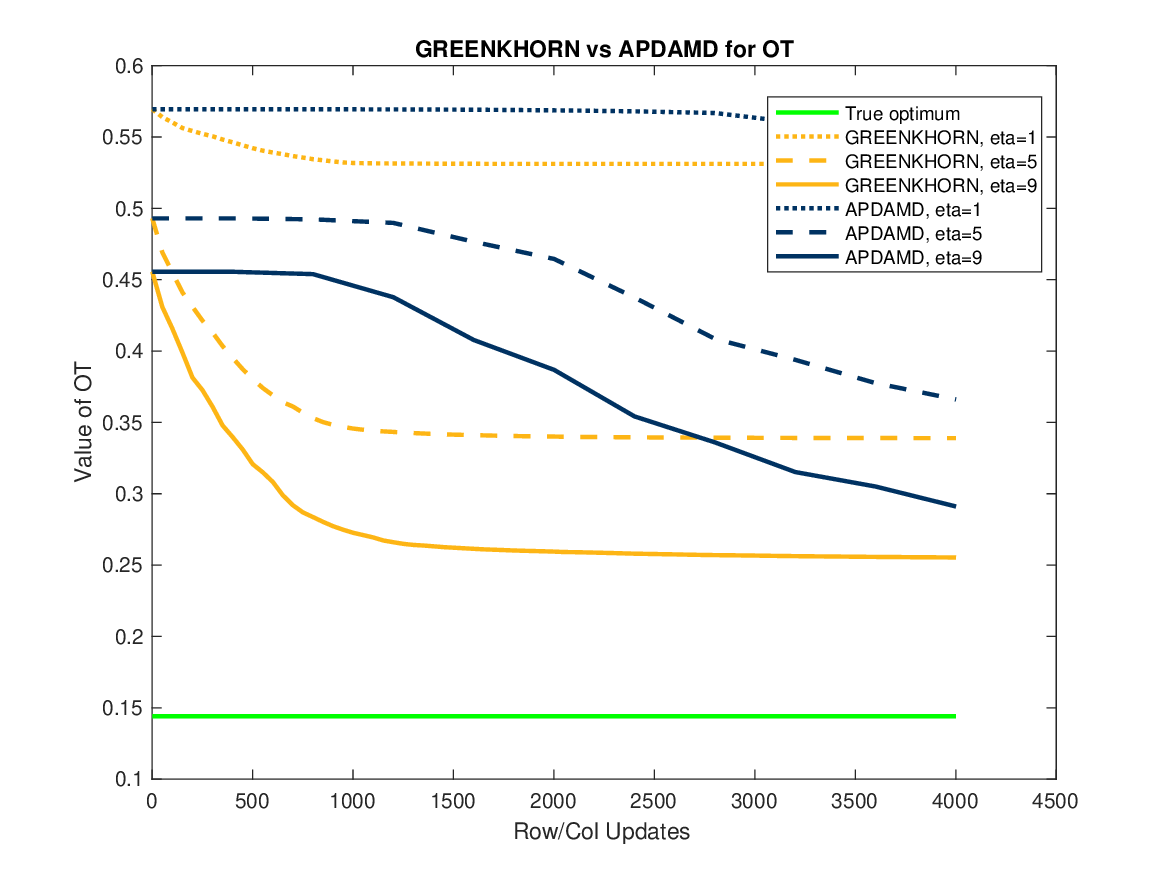}
\end{minipage}
\quad
\begin{minipage}[b]{.5\textwidth}
\includegraphics[width=75mm,height=55mm]{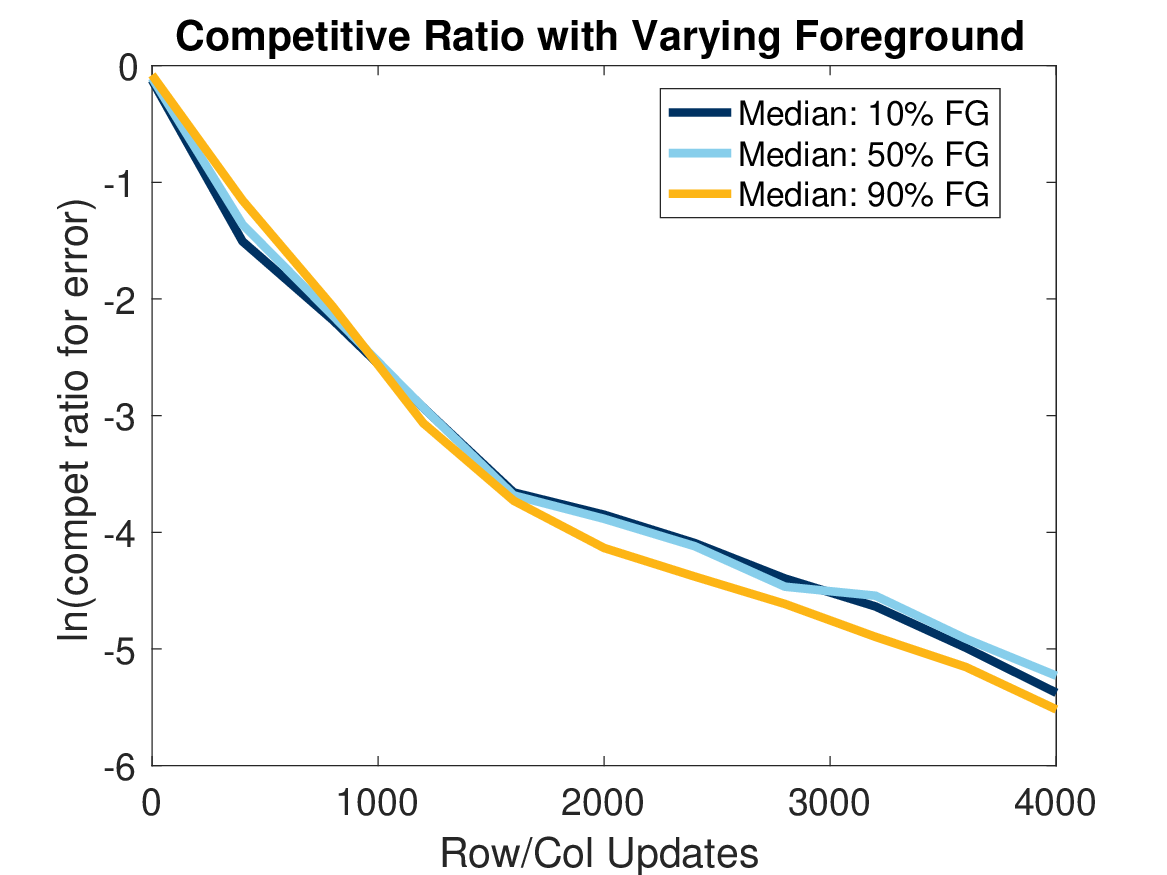}
\end{minipage}
\caption{Performance of the Greenkhorn and APDAMD algorithms in the synthetic images. All the four images correspond to that in Figure~\ref{fig:green_sink_synthetic}, showing that the Greenkhorn algorithm is faster and more robust than the APDAMD algorithm. Note that $\log(d(P_{GH})/d(P_{MD}))$ on 10 random pairs of images is consistently used, where $d(P_{GH})$ and $d(P_{MD})$ are for the Greenkhorn and APDAMD algorithms, respectively.}\label{fig:green_amd_synthetic}
\end{figure*}
\paragraph{Experimental setting:} We perform three pairwise comparative experiments: Sinkhorn versus Greenkhorn, APDAGD versus APDAMD, and Greenkhorn versus APDAMD by running these algorithms with ten randomly selected pairs of synthetic images. We also evaluate all the algorithms with varying regularization parameter $\eta \in \{1, 5, 9\}$ and the optimal value of the unregularized optimal transport problem, as suggested by~\citet{Altschuler-2017-Near}. 
\begin{figure*}[!t]
\begin{minipage}[b]{.3\textwidth}
\includegraphics[width=54mm,height=40mm]{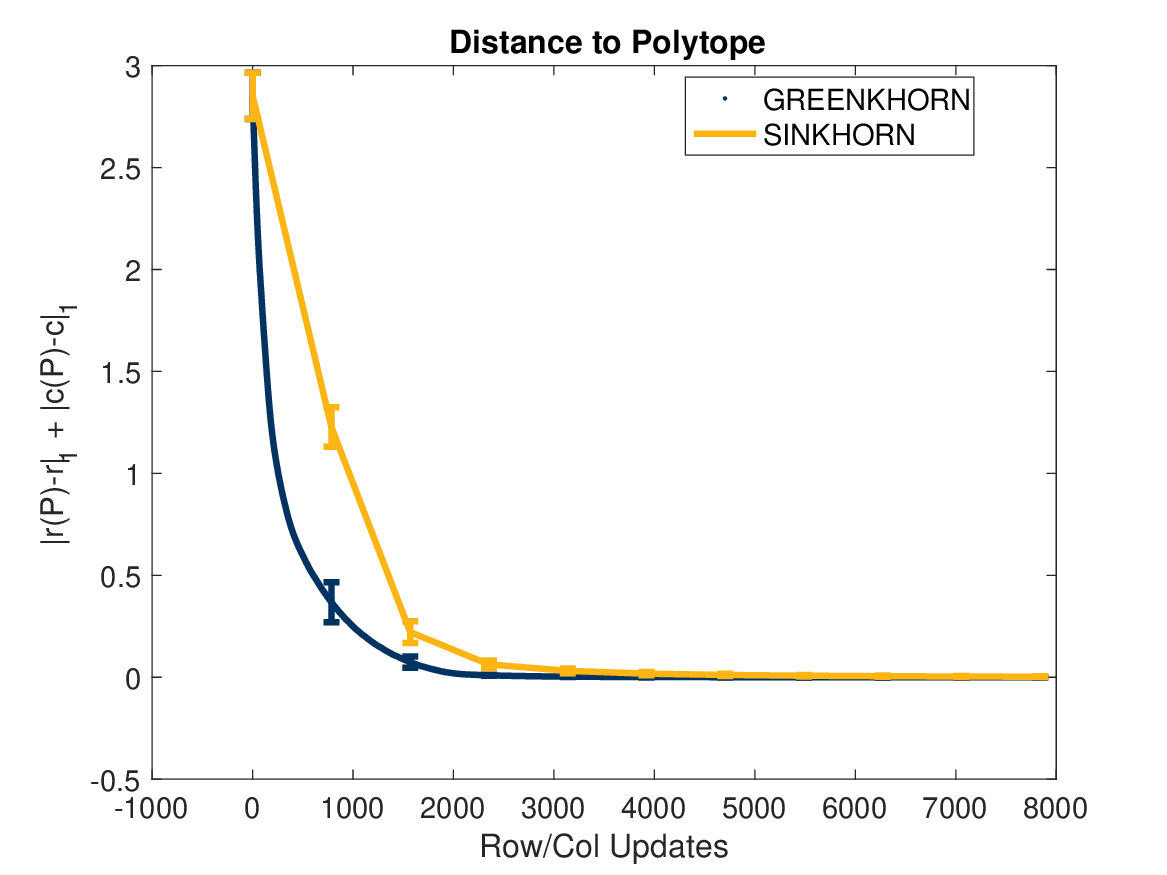}
\end{minipage}
\quad
\begin{minipage}[b]{.3\textwidth}
\includegraphics[width=54mm,height=40mm]{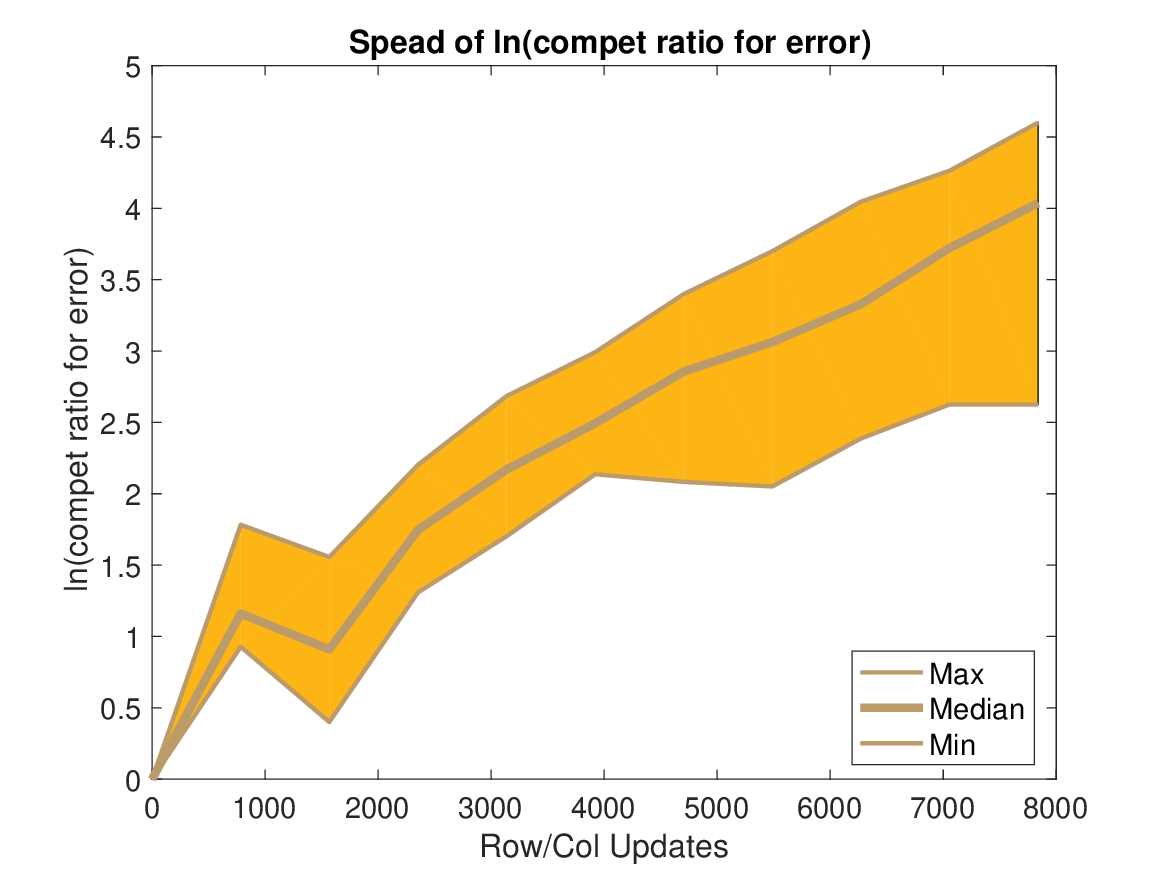}
\end{minipage}
\quad
\begin{minipage}[b]{.3\textwidth}
\includegraphics[width=54mm,height=40mm]{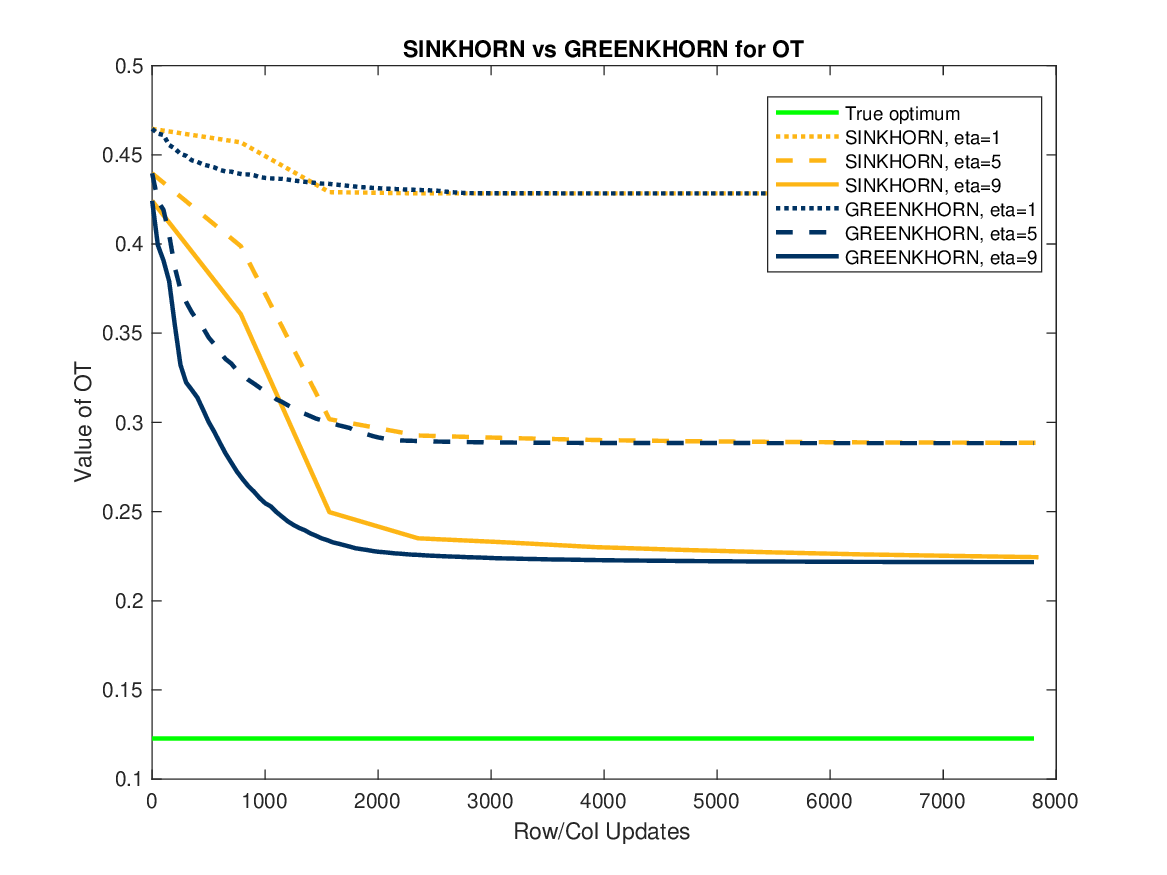}
\end{minipage}
\\
\begin{minipage}[b]{.3\textwidth}
\includegraphics[width=54mm,height=40mm]{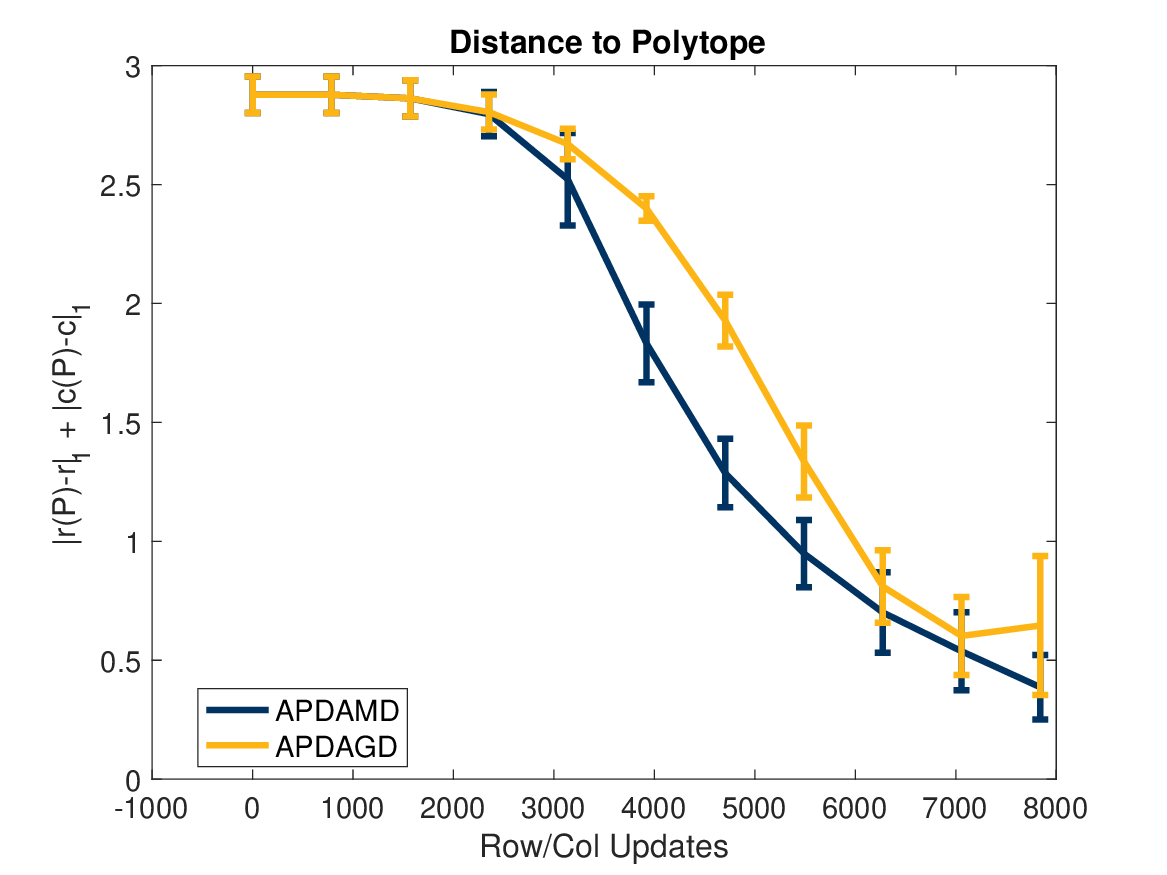}
\end{minipage}
\quad
\begin{minipage}[b]{.3\textwidth}
\includegraphics[width=54mm,height=40mm]{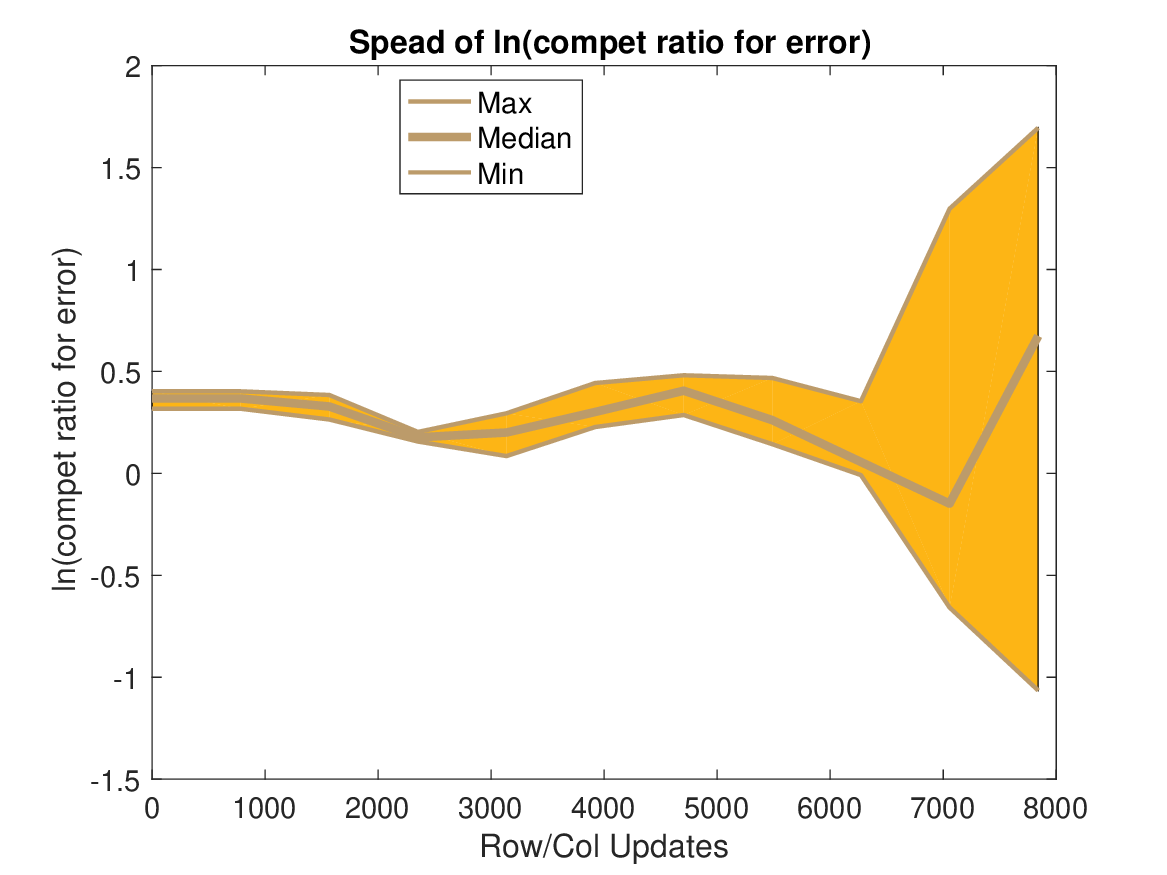}
\end{minipage}
\quad
\begin{minipage}[b]{.3\textwidth}
\includegraphics[width=54mm,height=40mm]{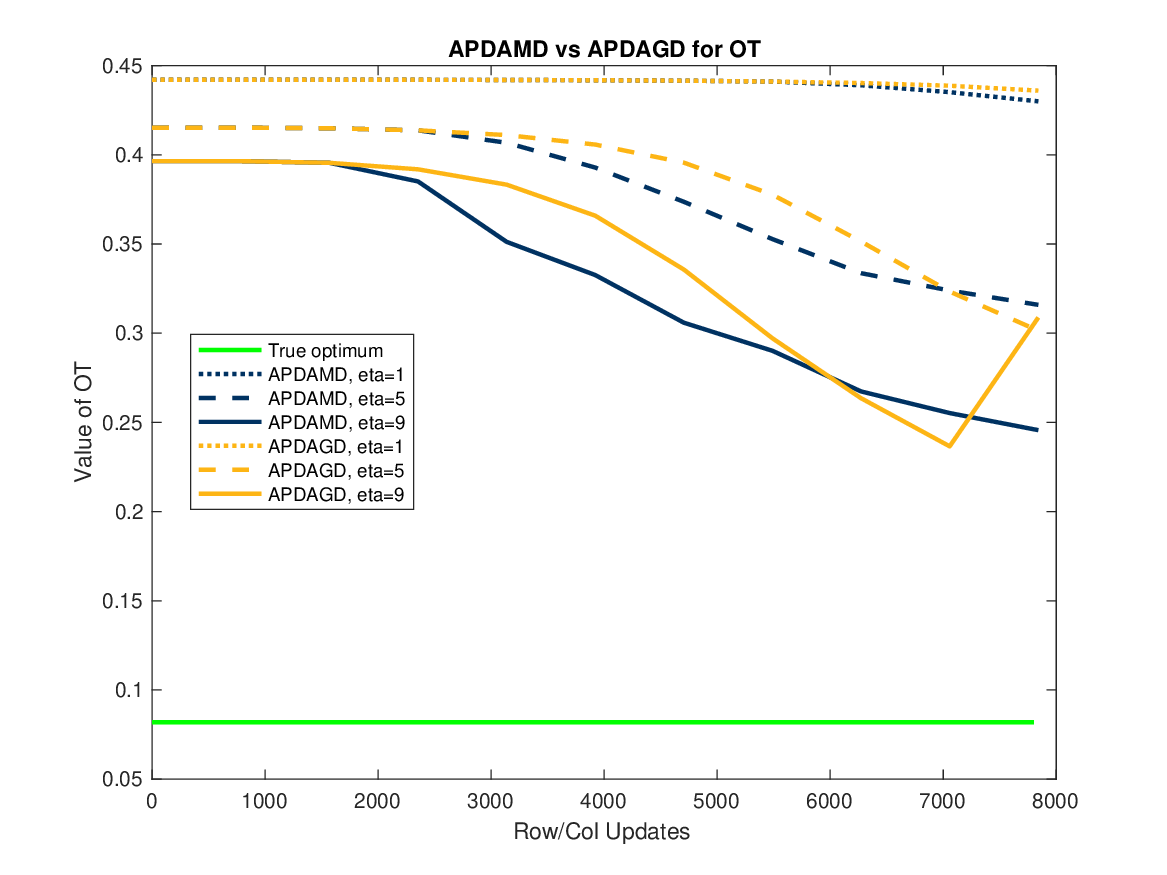}
\end{minipage} 
\\
\begin{minipage}[b]{.3\textwidth}
\includegraphics[width=54mm,height=40mm]{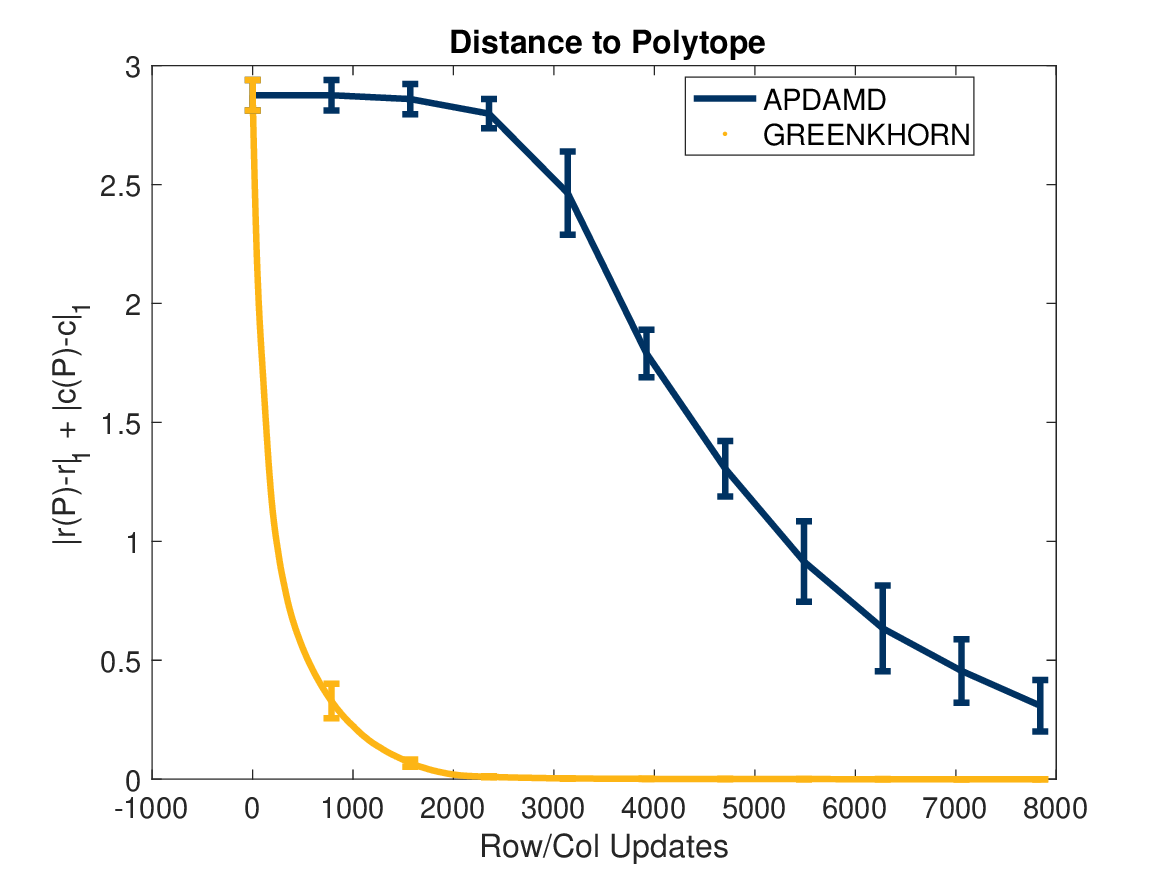}
\end{minipage}
\quad
\begin{minipage}[b]{.3\textwidth}
\includegraphics[width=54mm,height=40mm]{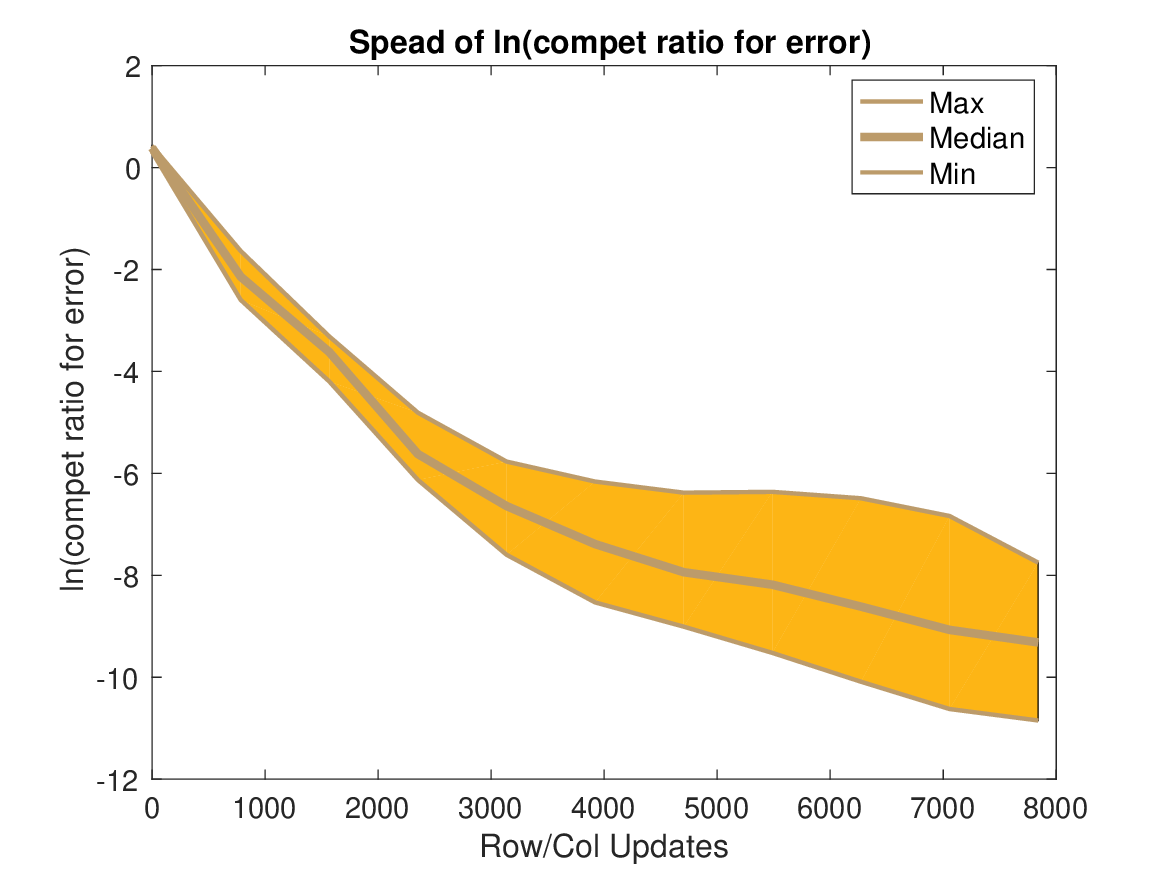}
\end{minipage}
\quad
\begin{minipage}[b]{.3\textwidth}
\includegraphics[width=54mm,height=40mm]{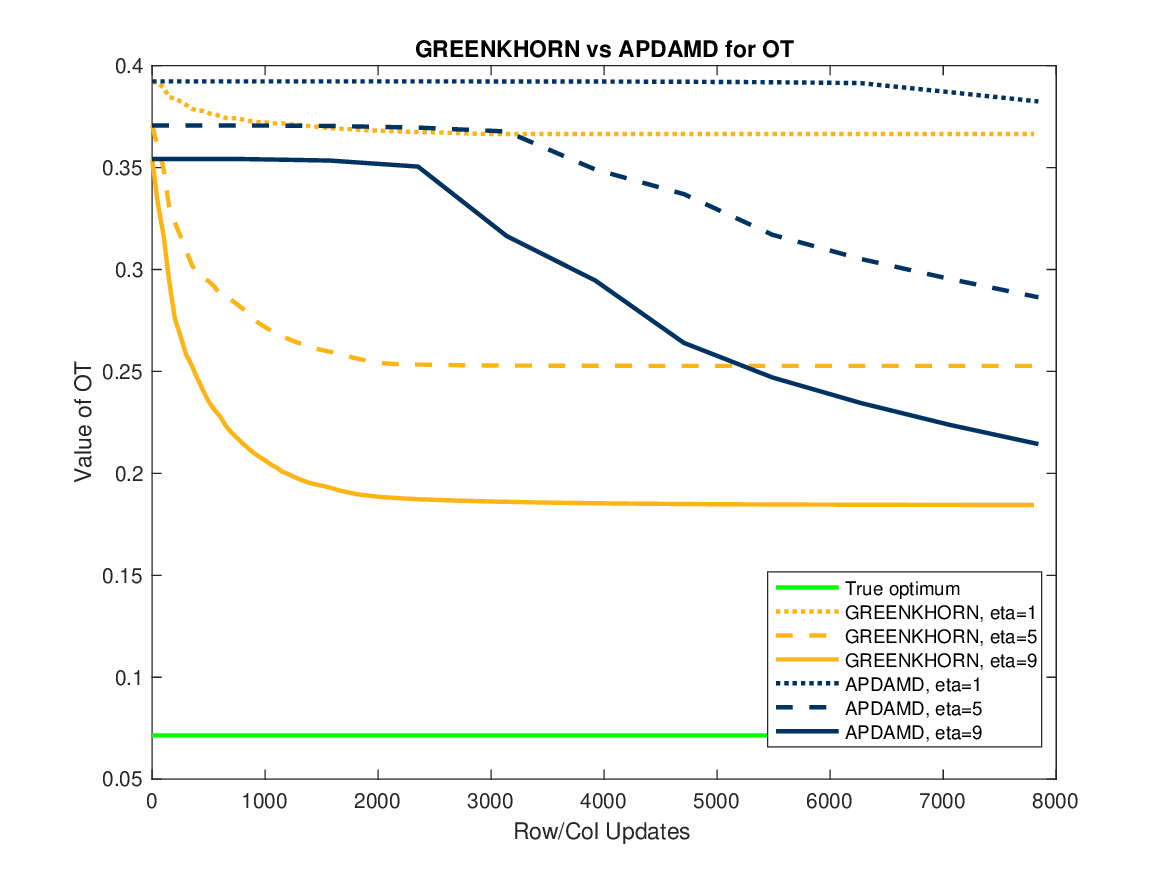}
\end{minipage}
\caption{Performance of the Sinkhorn, Greenkhorn, APDAGD and APDAMD algorithms on the MNIST real images. In the first row of images, we compare the Sinkhorn and Greenkhorn algorithms in terms of iteration counts. The leftmost image specifies the distances $d(P)$ to the transportation polytope for two algorithms; the middle image specifies the maximum, median and minimum of competitive ratios $\log(d(P_S)/d(P_G))$ on ten random pairs of MNIST images, where $P_S$ and $P_G$ stand for the outputs of APDAGD and APDAMD, respectively; the rightmost image specifies the values of regularized OT with varying regularization parameter $\eta \in \{1, 5, 9\}$. In addition, the second and third rows of images present comparative results for APDAGD versus APDAMD and Greenkhorn versus APDAMD. In summary, the experimental results on the MNIST images are consistent with that on the synthetic images. }
\label{fig:MNIST}
\end{figure*}
\paragraph{Experimental results:} We present the experimental results in Figure~\ref{fig:green_sink_synthetic}, Figure~\ref{fig:amd_agd_synthetic}, and Figure~\ref{fig:green_amd_synthetic} with different choices of
regularization parameters and different choices of coverage ratio of the foreground. Figure~\ref{fig:green_sink_synthetic} and~\ref{fig:green_amd_synthetic} show that the Greenkhorn 
algorithm performs better than the Sinkhorn and APDAMD algorithms in terms of iteration numbers, illustrating the improvement achieved by using greedy coordinate descent on the dual regularized OT problem. This also supports our theoretical assertion that the Greenkhorn algorithm has the complexity bound as good as the Sinkhorn algorithm (cf.\ Theorem~\ref{Theorem:Greenkhorn-Total-Complexity}). Figure~\ref{fig:amd_agd_synthetic} shows that the APDAMD algorithm with $\delta = n$ and the Bregman divergence equal to $(1/2n)\|\cdot\|^2$ is not faster than the APDAGD algorithm but is more robust. This makes sense since their complexity bounds are the same in terms of $n$ and $\varepsilon$ (cf. Theorem~\ref{Theorem:APDAMD-Total-Complexity} and Proposition~\ref{prop:correct_complex_APDAGD}). On the other hand, the robustness comes from the fact that the APDAMD algorithm can stabilize the training by using $\|\cdot\|_\infty$ in the line search criterion. 

\subsection{MNIST images}
We proceed to the comparison between different algorithms on real images, using essentially the same evaluation metrics as in the synthetic images. 
\paragraph{Image processing:} The MNIST dataset consists of 60,000 images of handwritten digits of size 28 by 28 pixels. To understand better the dependence on $n$ for our algorithms, we add a very small noise term ($10^{-6}$) to all the zero elements in the measures and then normalize them such that their sum becomes one.
\begin{figure*}[!ht]
\begin{minipage}[b]{.3\textwidth}
\includegraphics[width=55mm,height=40mm]{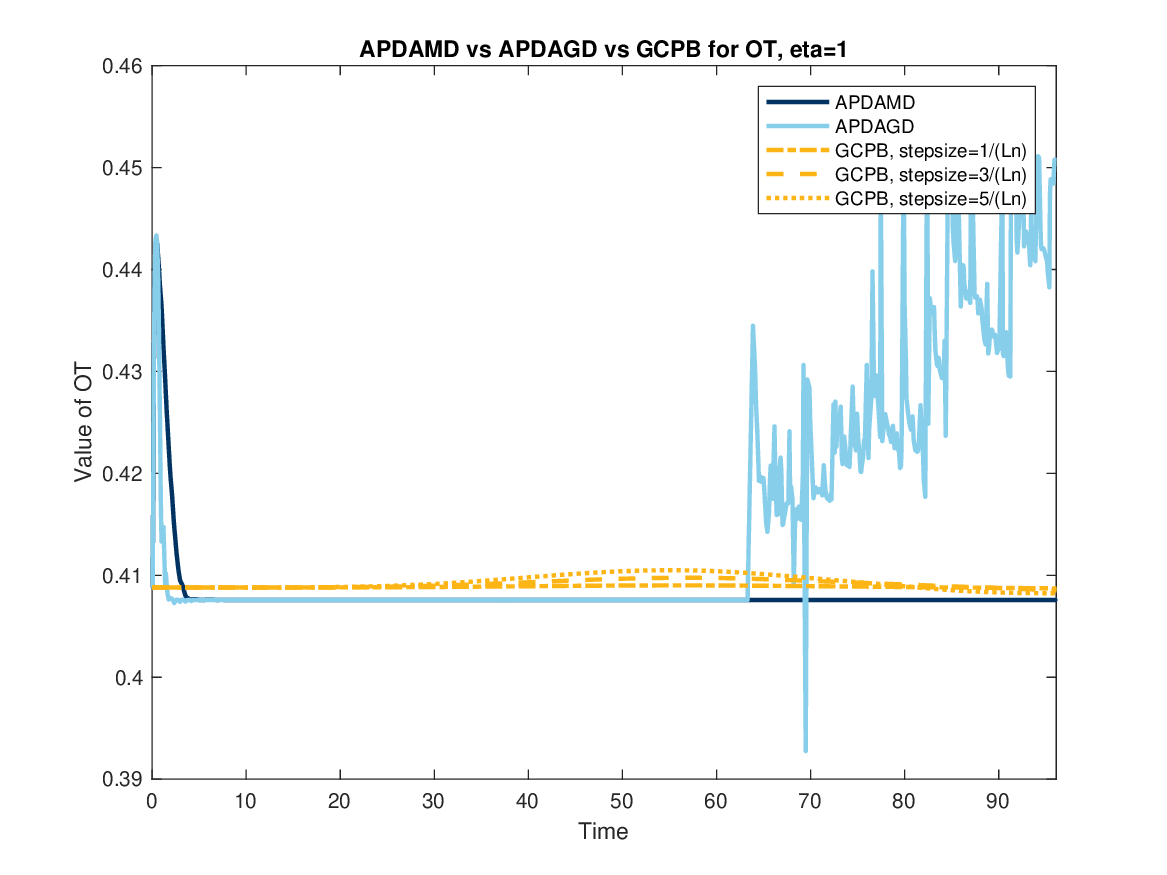}
\end{minipage}
\quad
\begin{minipage}[b]{.3\textwidth}
\includegraphics[width=55mm,height=40mm]{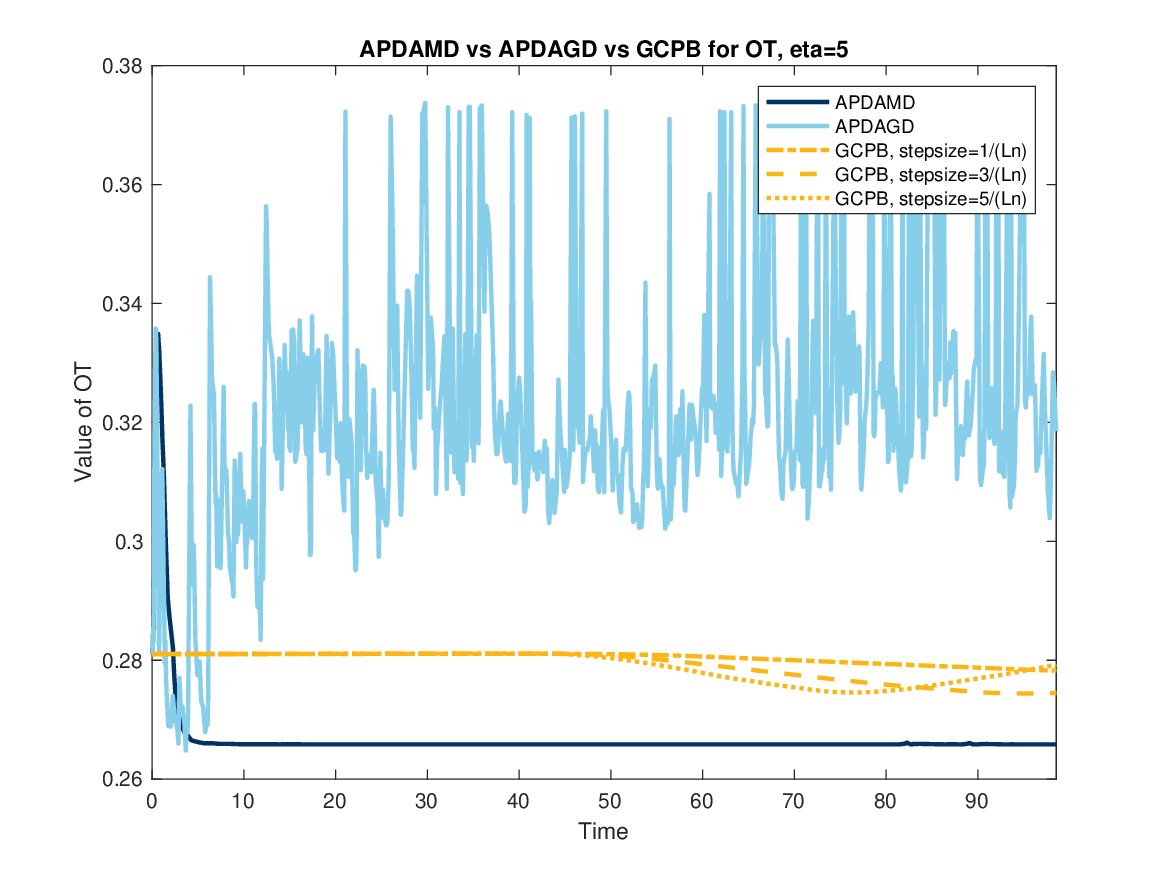}
\end{minipage}
\quad
\begin{minipage}[b]{.3\textwidth}
\includegraphics[width=55mm,height=40mm]{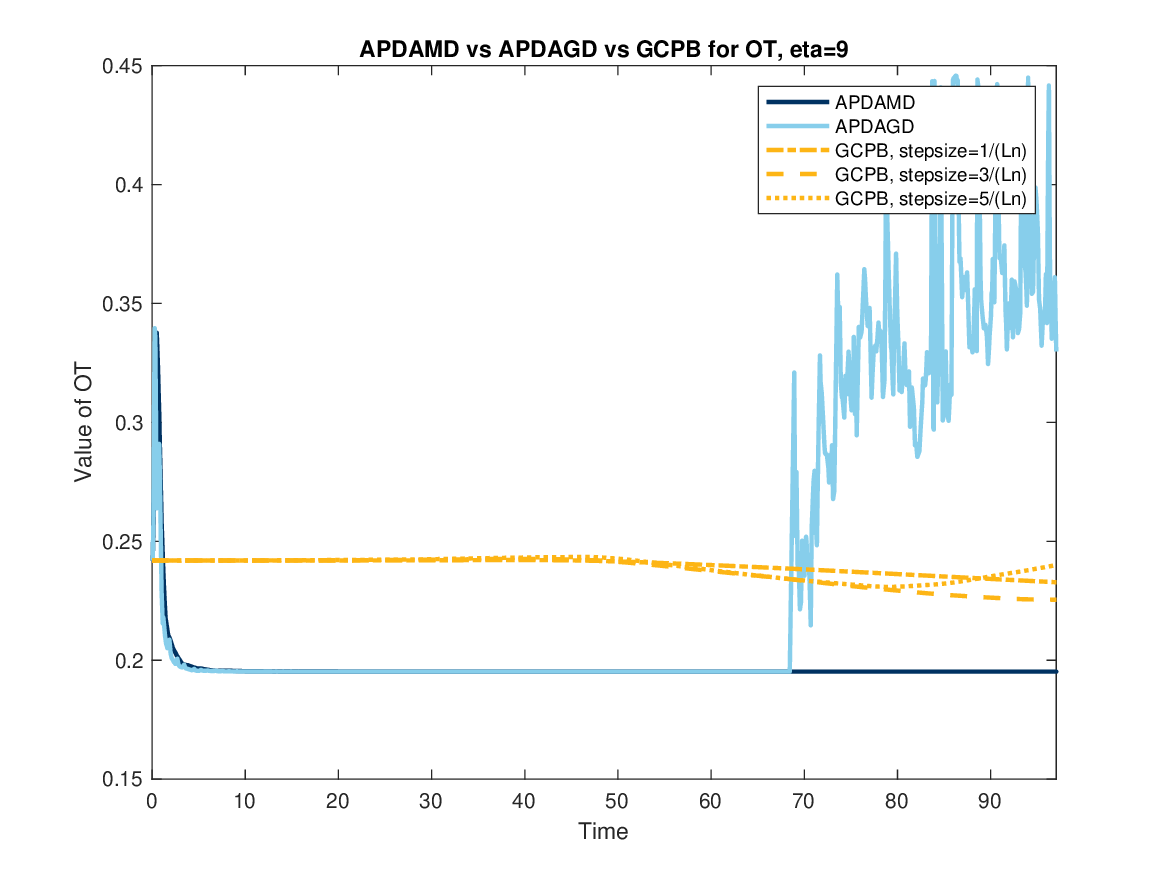}
\end{minipage}
\caption{Performance of the GCPB, APDAGD and APDAMD algorithms in term of time on the MNIST real images. These three images specify the values of regularized OT with varying regularization parameter $\eta \in \{1, 5, 9\}$, showing that the APDAMD algorithm is faster and more robust than the APDAGD and GCPB algorithms. }
\label{fig:ot-MNIST}
\end{figure*}
\paragraph{Experimental results:} We present the experimental results in Figure~\ref{fig:MNIST} and Figure~\ref{fig:ot-MNIST} with different choices of regularization parameters as well as the coverage ratio of the foreground on the real images. Figure~\ref{fig:MNIST} shows that the Greenkhorn algorithm is the fastest among all the candidate algorithms in terms of iteration count. Also, the APDAMD algorithm outperforms the APDAGD algorithm in terms of robustness and efficiency. All the results on real images are consistent with those on the synthetic images. Figure~\ref{fig:ot-MNIST} shows that the APDAMD algorithm is faster and more robust than the APDAGD and GCPB algorithms. We conclude that APDAMD algorithm has a more favorable performance profile than APDAGD algorithm for solving the regularized OT problem.

\section{Discussion}\label{sec:discussion}
We have provided detailed analyses of convergence rates for two algorithms for solving regularized OT problems. First, we established that the complexity bound of the Greenkhorn algorithm can be improved to $\bigOtil(n^2\varepsilon^{-2})$, which matches the best known complexity bound of the Sinkhorn algorithm.  We believe that this helps to explain why the Greenkhorn algorithm outperforms the Sinkhorn algorithm in practice. Second, we have proposed a novel adaptive primal-dual accelerated mirror descent (APDAMD) algorithm for solving regularized OT problems. We showed that the complexity bound of our algorithm is $\bigOtil(n^2\sqrt{\delta}\varepsilon^{-1})$, where $\delta$ is the inverse of the strongly convex module of the Bregman divergence with respect to $\|\cdot\|_\infty$. Finally, we pointed out that an existing complexity bound for the APDAGD algorithm from the literature is not valid in general by providing a concrete counterexample. We instead established a complexity bound for the APDAGD algorithm is $\bigOtil(n^{5/2}\varepsilon^{-1})$, by exploiting the connection between the APDAMD and APDAGD algorithms.

There are many interesting directions for further research. First, the complexity bound of the APDAMD algorithm heavily depends on $\delta$. As we mentioned earlier, a simple upper bound for $\delta$ is the dimension $n$. However, this results in a complexity bound for the APDAMD algorithm of $\bigOtil(n^{5/2}\varepsilon^{-1})$, which is unsatisfactory. It is of significant theoretical interest to investigate whether we can improve the dependence of $\delta$ on $n$, such as $n^\tau$ for some $\tau < 1$. Another possible direction is to extend the APDAMD algorithm to the computation of Wasserstein barycenters. That computation has been proposed for a variety of applications in machine learning and statistics~\citep{Ho-2018-Probabilistic,Srivastava-2015-WASP, Srivastava-2018-Scalable}, but its theoretical understanding is limited despite recent developments in fast algorithms for solving the problem~\citep{Cuturi-2014-Fast, Dvurechenskii-2018-Decentralize}.

\section*{Acknowledgments}
We would like to thank Pavel Dvurechensky, Alexander Gasnikov, and Alexey Kroshnin for helpful discussion with the complexity bounds of APDAMD and APDAGD algorithms. We would like to thank anonymous referee for helpful comments on Lemma~\ref{Lemma:dualOT-smoothness} and Eq.~\eqref{prob:dualregOT-APDAMD}. This work was supported in part by the Mathematical Data Science program of the Office of Naval Research under grant number N00014-18-1-2764.
\bibliographystyle{plainnat}
\bibliography{ref}


\end{document}